\def\draft{0} 

\documentclass{article}

\PassOptionsToPackage{numbers}{natbib}

\usepackage[final]{neurips_2022}

\usepackage[utf8]{inputenc} 
\usepackage[T1]{fontenc}    
\usepackage{url}            
\usepackage{hyperref}       

\usepackage{mathtools}
\usepackage{amsmath}
\usepackage{amssymb}
\usepackage{amsthm}
\usepackage{color}
\usepackage{comment}
\usepackage{dsfont}
\usepackage{bm}
\usepackage{colortbl}
\usepackage{booktabs}
\usepackage{multirow}
\usepackage{xcolor}
\usepackage{wrapfig}

\usepackage[ruled,vlined,linesnumbered,noend]{algorithm2e}
\usepackage{algpseudocode}

\usepackage{microtype}
\usepackage{graphicx}
\usepackage{subfigure}

\newcommand{\rebuttal}[1]{{\color{black} #1}}

\theoremstyle{plain}
\newtheorem{theorem}{Theorem}[section]

\newtheorem{lemma}[theorem]{Lemma}

\theoremstyle{definition}
\newtheorem{definition}[theorem]{Definition}

\theoremstyle{remark}
\newtheorem{remark}[theorem]{Remark}

\newtheorem{innercustomthm}{Theorem}
\newenvironment{customthm}[1]
  {\renewcommand\theinnercustomthm{#1}\innercustomthm}
  {\endinnercustomthm}

\newcommand{\E}{\mathbb{E}}

\newcommand{\R}{\mathbb{R}}
\newcommand{\eps}{\varepsilon}
\newcommand{\D}{\mathcal{D}}
\newcommand{\norm}[1]{\left\lVert#1\right\rVert}
\newcommand{\N}{\mathcal{N}}

\newcommand{\M}{\mathcal{M}}
\DeclareMathOperator*{\argmin}{\arg\!\min}

\newcommand{\Ind}{\mathds{1}}
\newcommand{\g}{\nabla}
\newcommand{\iden}{\bm{1}}
\newcommand{\inner}[2]{\left\langle #1, #2 \right\rangle}

\newcommand{\Renyi}{R\'enyi~}
\newcommand{\localSen}{\mathbb{LS}}
\newcommand{\globalSen}{\mathbb{GS}}
\newcommand{\lap}{\mathrm{Lap}}
\newcommand{\ptr}{\mathrm{PTR}}

\newcommand{\rdpfunc}{f_\alpha}

\newcommand{\poisson}{\mathrm{PoissonSample}}

\newcommand{\cX}{\mathcal{X}}
\newcommand{\cY}{\mathcal{Y}}
\newcommand{\MS}{\mathrm{MultiSets}}

\newcommand{\tsgd}{\mathrm{TSGD}}
\newcommand{\gau}{\mathrm{Gaussian}}

\ifnum\draft=1
\newcommand{\full}[1]{\textcolor{green}{[Full Version: #1]}}
\else
\newcommand{\full}[1]{\null}
\fi

\ifnum\draft=1
\newcommand{\add}[1]{\textcolor{blue}{#1}}
\else
\newcommand{\add}[1]{#1}
\fi

\newcommand{\epslap}{\eps_{\lap}}
\newcommand{\epsgau}{\eps_{\N}}

\newcommand{\epsRlap}{\eps_{\mathrm{R-}\lap}}
\newcommand{\epsRgau}{\eps_{\mathrm{R-}\N}}

\title{\Renyi Differential Privacy of Propose-Test-Release and Applications to Private and Robust Machine Learning}

\author{%
  Jiachen T. Wang \\
  Princeton University\\
  \texttt{tianhaowang@princeton.edu} \\
  \And
  Saeed Mahloujifar \\
  Princeton University \\
  \texttt{sfar@princeton.edu} \\
  \AND
  Shouda Wang \\
  Princeton University\\
  \texttt{sw1041@princeton.edu} \\
  \And
  Ruoxi Jia \\
  Virginia Tech \\
  \texttt{ruoxijia@vt.edu} \\
  \And
  Prateek Mittal \\
  Princeton University\\
  \texttt{pmittal@princeton.edu} \\
}

\begin{document}

\maketitle

\begin{abstract}
Propose-Test-Release (PTR) is a differential privacy framework that works with local sensitivity of functions, instead of their global sensitivity. This framework is typically used for releasing robust statistics such as median or trimmed mean in a differentially private manner. 
\add{
While PTR is a common framework introduced over a decade ago, using it in applications such as robust SGD where we need many adaptive robust queries is challenging. This is mainly due to the lack of \Renyi Differential Privacy (RDP) analysis, an essential ingredient underlying the moments accountant approach for differentially private deep learning.}
In this work, we generalize the standard PTR and derive the first RDP bound for it when the target function has bounded global sensitivity. We show that our RDP bound for PTR yields tighter DP guarantees than the directly analyzed $(\eps, \delta)$-DP. 
We also derive the algorithm-specific privacy amplification bound of PTR under subsampling. We show that our bound is much tighter than the general upper bound and close to the lower bound. Our RDP bounds enable tighter privacy loss calculation for the composition of many adaptive runs of PTR. As an application of our analysis, we show that PTR and our theoretical results can be used to design differentially private variants for byzantine robust training algorithms that use robust statistics for gradients aggregation. We conduct experiments on the settings of label, feature, and gradient corruption across different datasets and architectures. We show that PTR-based private and robust training algorithm significantly improves the utility compared with the baseline. 

\end{abstract}

\section{Introduction}
\label{sec:intro}

Privacy is a major concern for deploying machine learning (ML). In response, differential privacy (DP) \cite{dwork2006calibrating} has become the de-facto measure of privacy. For a differentially private mechanism, the probability distribution of the mechanism's outputs on a dataset should be close to the distribution of its outputs on the same dataset with any single individual’s data replaced. A general recipe for releasing the value of a function $f$ on dataset $S$ in a differentially private way is adding random noise to $f(S)$ (output perturbation), where noise magnitude should scale with $f$'s global sensitivity. 


However, it might be over-conservative to add noise scaled with global sensitivity. There is a line of research on whether we can do better (e.g., \cite{nissim2007smooth, dwork2009differential, thakurta2013differentially, kasiviswanathan2013analyzing}). 
Propose-Test-Release (PTR) \cite{dwork2009differential} is a framework that improves the general recipe with the notion of \emph{local sensitivity}. 
The main idea of PTR is as follows: instead of adding noise with respect to global sensitivity, we propose an amount of noise that is tolerable for many common queries. When we receive the actual query, we test (in a differentially private way) whether answering the query with the proposed amount of noise is enough for privacy. If the proposed noise is too small for limiting the privacy loss from the actual query, we may refuse to answer the query or respond with a larger noise. 
PTR works especially well for releasing robust statistics such as median or trimmed mean, as the robust statistics usually have small local sensitivity on most common inputs.


While PTR is a basic framework that ages back to the early days after the introduction of DP, it has not been used for differentially private optimization before. DP-SGD \cite{DBLP:conf/ccs/AbadiCGMMT016} is the general backbone for differentially private deep learning and optimization. 
One major challenge of augmenting SGD with PTR and training DP models is the calculation of privacy parameters after a large number of adaptive compositions. 
\Renyi differential privacy (RDP) and Moment Accountant \cite{mironov2017renyi, DBLP:conf/ccs/AbadiCGMMT016} enable us to calculate tighter privacy parameters for training DP models. Without the RDP bound of PTR, we need to calculate the privacy parameters using the advanced composition theorems \cite{dwork2014algorithmic} that can lead to significantly looser privacy bounds. 
Besides, we also need the bound of \emph{privacy amplification by subsampling} for PTR, which is the other important support for training DP models. It allows us to exploit the stochasticity of SGD for the interest of stronger privacy guarantees.

\textbf{Technical Overview. }
In this work, we derive the \Renyi DP bound for PTR, as well as for its Poisson subsampled variant when the target function has bounded global sensitivity. Our bounds make it possible for us to use PTR framework in augmenting private SGD. PTR could be characterized by three mechanisms. The first mechanism is $\M_1$ that determines which mechanism to run next. Depending on the outcome of $\M_1$ we then either run $\M_2$ or $\M_2'$. 
It is often the case that the worst-case privacy loss of one of $\M_2$ or $\M_2'$ is much larger than the other. Our bound exploits the fact that the worst-case will happen with small probability.
Specifically, instead of considering the worst-case privacy loss between $\M_2$ and $\M_2'$, and naively composing it with $\M_1$, we show that for RDP we can tighten the bound by the \emph{average} privacy loss of $\M_2$ and $\M'_2$ under the distribution imposed by $\M_1$. 
Direct $(\eps, \delta)$-DP analysis does not enjoy this benefit; compared with directly analyzing the $(\eps, \delta)$-DP bound of PTR, we show that by first bounding the RDP of PTR and then convert it to $(\eps, \delta)$-DP can lead to better privacy guarantee. 
Our proof could serve as a general recipe for analyzing DP/RDP guarantees for composed mechanisms where the privacy loss of each mechanism is adaptively determined. 
Additionally, we extend our analysis to the RDP of \emph{subsampled} PTR. Our algorithm-specific analysis (the ``white-box bound'') allows us to get tighter privacy amplification bounds, compared with the one obtained by general subsampled RDP bound that supports any mechanisms (the ``black-box bound'') \cite{zhu19poisson}.
The proof tackled several additional difficulties compared with the analysis for simple Gaussian mechanism \citep{mironov2019r}. 
By numerical verification, we show that our RDP bound for subsampled PTR is much tighter than the black-box bound, and is close to the lower bound (Figure \ref{fig:rdp-compare}.).



\textbf{Applications of PTR in ML.} 
\add{
PTR is especially suitable for improving the utility of privatizing robust statistics such as trimmed mean, as these functions usually have a much smaller local sensitivity compared to their global sensitivity. A critical application of robust statistics for machine learning is defending against corrupted data.}
The learning algorithm should be robust in the presence of corrupted data (referred to as \emph{Byzantine failure} \cite{lamport1982byzantine}). 
Since privacy and robustness are two major concerns for ML training, developing techniques that achieve both goals simultaneously is desirable. We demonstrate the application of PTR in incorporating differential privacy with robust SGD methods. 
\add{A popular} technique for robustifying SGD is to replace the mean with some robust statistics (e.g., trimmed mean) for gradient aggregation \cite{yin2018byzantine, acharya2021robust, gupta2021byzantine}. 
We use trimmed mean as an example of showing how to augment robust SGD with DP through PTR. We show that the augmented SGD still maintains the robustness guarantee. We conduct extensive experiments on defending against three kinds of attacks: label, feature, and gradient corruptions, and we show that PTR-based robust SGD achieves much better utility than naively privatizing the robust SGD with global sensitivity.

\section{Related Work}
\label{sec:relatedwork}

Since the seminal work of \cite{dwork2009differential}, the Propose-Test-Release has become a common DP framework mainly used in statistical inference \citep{brunel2020propose, liu2021differential}. 
To the best of our knowledge, this work is the first application of PTR in machine learning. 

\add{
Several prior works (e.g., \cite{li2012sampling, balle2018privacy, wang2019subsampled, zhu19poisson}) focus the general subsampled DP/RDP bound that supports any mechanisms. \cite{wang2019subsampled} derives a general RDP privacy amplification for ``sampling without replacement'' scheme, and \cite{zhu19poisson} obtains a similar result under Poisson subsampling (i.e., including each data point independently at random with certain probability). Only a small body of recent works study algorithm-specific privacy amplification bound by subsampling. However, most of them focus on subsampled Gaussian mechanism \cite{DBLP:conf/ccs/AbadiCGMMT016, bun2018composable, mironov2019r}. This work derives the first subsampled RDP bound specific for flexible algorithms such as PTR.}

\add{Designing machine learning and optimization algorithms that achieve both privacy and Byzantine-robustness is certainly an important direction. Nevertheless, there are only few works on this line so far.} 
\cite{guerraoui2021combining} considered the problem of achieving privacy and byzantine resilience in distributed SGD with an untrusted server. The privacy level they considered is essentially local differential privacy, which is orthogonal to our focus. \cite{he2020secure} and \cite{so2020byzantine} aim for both robustness and secure multiparty computation instead of differential privacy. 
\cite{ma2019data} naively applies differential privacy to defend against data poisoning attack. However, they find that DP alone cannot defend adversaries that poison a large fraction of training examples. 
Aside from the setting of training ML models, \cite{liu2021robust} propose a polynomial time algorithm that achieves both goals for mean estimation via privatizing filter-based robust mean estimator \cite{diakonikolas2017being}. \cite{esfandiari2021tight} develop a robust and differentially private mean estimator based on exponential mechanism. However, both of their approaches become inefficient (even in polynomial time asymptotically) in high dimensional settings.

\section{Background}
\label{sec:preliminary}


In this section, we introduce some background on differential privacy, \Renyi differential privacy, and privacy-amplification by subsampling. We will also introduce notations as we proceed.

\textbf{Differential Privacy. }
Differential privacy is a framework for protecting privacy when performing statistical releases on a dataset with sensitive information about individuals (see the surveys \cite{dwork2014algorithmic,vadhan2017complexity}). Specifically, for a differentially private mechanism, the probability distribution of the mechanism's outputs of a dataset should be close to the distribution of its outputs on the same dataset with any single individual’s data replaced. To formalize this, we call two datasets $S$, $S'$, each multisets over a data universe $\cX$, {\em adjacent} if one can be obtained from the other by \emph{adding or removing} a single element of $\cX$. Further, we use $d(S, S')$ to denote the number of times of adding/removing of data points to transform $S$ to $S'$. So $S$ and $S'$ are adjacent if and only if $d(S, S') = 1$. 

\begin{definition}[Differential Privacy \cite{dwork2006calibrating}] \label{def:DP}
For $\eps, \delta \ge 0$, a randomized algorithm $\M : \MS(\cX)\rightarrow \cY$ is 
{\em $(\eps, \delta)$-differentially private} if for every dataset pair $S, S'\in \MS(\cX)$ such that $d(S, S')=1$, we have: 
\begin{equation} \label{req:approxDP}
\forall\ T\subseteq \cY\ \Pr[\M(S) \in T] \le e^\eps\cdot \Pr[\M(S') \in T] + \delta
\end{equation}
where the randomness is over the coin flips of $\M$. 
\end{definition}

\textbf{\Renyi Differential Privacy (RDP).}
\Renyi differential privacy (RDP) is a variant of the standard $(\eps, \delta)$-DP that uses R{\'e}nyi-divergence as a distance metric \add{between the output distributions of $\M(S)$ and $\M(S')$}, which is particularly useful in training differentially private machine learning models. 

\begin{definition}[\Renyi Differential Privacy \cite{mironov2017renyi}]
We say that a mechanism $\M$ is $(\alpha, \eps)$-RDP with order $\alpha \in (1, \infty)$ if for every dataset pair $S, S'\in \MS(\cX)$ such that $d(S, S')=1$, we have: 
\begin{align}
D_{\alpha}\left(\M(S) \|  \M\left(S^{\prime}\right)\right) :=\frac{1}{\alpha-1} \log \E_{o \sim \M\left(S^{\prime}\right)}\left[\left(\frac{\mu_{\M(S)}(o)}{\mu_{\M\left(S^{\prime}\right)}(o)}\right)^{\alpha}\right] \leq \eps
\end{align}
where $\mu_\M(\cdot)$ denotes the density function of $\M$'s distribution. Further, we denote the moment $E_\alpha \left(\M(S) \|  \M\left(S^{\prime}\right)\right) := \E_{o \sim \M\left(S^{\prime}\right)}\left[\left(\frac{\mu_{\M(S)}(o)}{\mu_{\M\left(S^{\prime}\right)}(o)}\right)^{\alpha}\right]$ and function $\rdpfunc(\eps) := \exp( (\alpha-1)\eps )$.
\end{definition}
As we can see, $(\alpha, \eps)$-RDP is essentially an upper bound for the moment
$
E_\alpha \left(\M(S) \|  \M\left(S^{\prime}\right)\right) \le \rdpfunc(\eps)
$
for all adjacent $S, S'$, where $\eps$ can be viewed as a degree of the privacy loss incurred by running $\M$.
A different $\alpha$ typically leads to a different privacy bound $\eps$. Following the convention of literature \citep{zhu19poisson}, we view $\eps$ as a function of $\alpha$, and the notation $\eps_\M(\alpha)$ means the algorithm $\M$ obeys $(\alpha, \eps_\M(\alpha))$-RDP. 
We note that RDP reduces to $(\eps, 0)$-DP when we take $\alpha = \infty$. 
The current tightest $(\eps, \delta)$-DP/RDP transformation we are aware of is by \cite{asoodeh2021three}.
A central property of DP/RDP is its behavior under composition. If we run multiple distinct differentially private algorithms on the same dataset, the resulting composed algorithm is also differentially private, with some degradation in the privacy parameters $(\eps, \delta)$. 
Specifically, if we run $k$ sequentially chosen $(\eps, \delta)$-DP algorithm on a dataset, the overall composed privacy parameter is $\left(\tilde{O}(\sqrt{k} \eps), k \delta+\delta^{\prime}\right)$-DP by Advanced composition theorem \citep{dwork2014algorithmic}. However, the Advanced composition theorem for $(\eps, \delta)$-DP is loose. 
On the contrary, the composition is trivial for RDP as $\eps_{\M_1 \circ \M_2}(\cdot) = \eps_{\M_1}(\cdot) + \eps_{\M_2}(\cdot)$. 
The Moment Accountant technique, which composes RDP and then transforms to DP, is a much simpler approach and often produces much more favorable privacy parameters than directly composing $(\eps, \delta)$-DP. Therefore, RDP and Moment Accountant are widely used to calculate the privacy guarantee in training differentially private deep learning models.

\textbf{Privacy amplification by subsampling.}
``Privacy amplification by subsampling'' is the other booster besides RDP / moments accountant that drives much of the recent progress in differentially private deep learning. 
Most of the existing works focus on Poisson subsampling \cite{balle2018privacy}, which samples each data point independently with a given sampling probability $q$. 
The tightest privacy amplification bound we are aware of for general mechanism under Poisson Subsampling is from \cite{zhu19poisson}. 

\section{Propose-Test-Release and \Renyi Differential Privacy}
\label{sec:main-theory}

In this section, we introduce a general version of the Propose-Test-Release framework, present our main results on its RDP bound, and its privacy amplification bound under Poisson subsampling. 


In our presentation, we use $\globalSen_f = \sup_{S, S': d(S, S')=1} \norm{f(S)-f(S')}$ to denote the global sensitivity of target function $f$ (in $\ell_2$ distance), and we use $\localSen_f(S) = \sup_{S, S': d(S, S')=1} \norm{f(S)-f(S')}$ to denote the local sensitivity of function $f$ on dataset $S$. Before we present our main results, we would like to remind the readers about the definition of Laplace and Gaussian mechanisms, as well as their DP/RDP bound, which will be referred to later in the main results. 

\add{
\textbf{Laplace Mechanism:} If $f$'s output is 1-dimensional, the Laplace mechanism $\M_{\lap, b}(S) = f(S) + \lap(0, b)$ obeys $(\epslap^{(\Tilde{b})}, 0)$-DP for $\epslap^{(\Tilde{b})} = e^{1/\Tilde{b}}$, and obeys $(\alpha, \epsRlap^{(\Tilde{b})}(\alpha))$-RDP for 
$
\epsRlap^{(\Tilde{b})}(\alpha) = \frac{1}{\alpha-1} \log \left( \frac{\alpha}{2\alpha-1}\exp \left(\frac{\alpha-1}{\Tilde{b}}\right) + \frac{\alpha-1}{2\alpha-1} \exp \left(-\frac{\alpha}{\Tilde{b}}\right)\right)
$, where $\Tilde{b} = b/\globalSen_f$ (the ``noise-to-sensitivity'' ratio). 
}

\add{
\textbf{Gaussian Mechanism:}
If $f$'s output is $d$-dimensional, the Gaussian mechanism $\M_{\N, \sigma}(S) = f(S) + \N(0, \sigma^2 \iden_d)$ obeys $(\eps_{\N}^{(\Tilde{\sigma})}(\delta), \delta)$-DP for $\eps_{\N}^{(\Tilde{\sigma})}(\delta) = \Tilde{\sigma}\sqrt{2\log(1.25/\delta)}$, and obeys $(\alpha, \epsRgau^{(\Tilde{\sigma})}(\alpha))$-RDP for $\epsRgau^{(\Tilde{\sigma})}(\alpha) = \frac{\alpha}{2 \Tilde{\sigma}^2 }$, where $\Tilde{\sigma} = \sigma/\globalSen_f$ (the ``noise-to-sensitivity'' ratio).
}

\subsection{Propose-Test-Release} 
\begin{wrapfigure}{R}{0.55\textwidth}
\centering
\begin{minipage}{0.5\textwidth}
\begin{algorithm}[H]
\SetAlgoLined
\SetKwInOut{Input}{input}
\SetKwInOut{Output}{output}
\Input{
$S$ -- dataset, 
$f_1$ -- target function, 
$f_2$ -- robust statistic, 
$\tau$ -- proposed local sensitivity bound of $f_2$, 
$\sigma_1, \sigma_2, b$ -- Gaussian/Laplace noise scales ($\sigma_1 > \sigma_2$), 
$\delta_0$ -- failure probability. 
}

$\Delta \leftarrow \min_{\Tilde{S} \in \{ \Tilde{S}: \localSen_{f_2}( \Tilde{S} ) > \tau \} } d \left(S, \Tilde{S} \right)$. 

$\widehat \Delta \leftarrow \Delta + \lap(0, b)$. 

\If{$\widehat \Delta \le \log(1/(2\delta_0)) b$}{
    \Return{$f_1(S) + \N(0, \sigma_1^2\iden_d)$}
}\Else{
    \Return{$f_2(S) + \N(0, \sigma_2^2\iden_d)$}
}

\caption{Propose-Test-Release with Laplace and Gaussian mechanism.}
\label{alg:ptr-general}
\end{algorithm}
\end{minipage}
\end{wrapfigure}
Naive use of Laplace/Gaussian mechanism may result in the poor utility of function output, as the global sensitivity of the target function may be intolerably large due to some extreme cases. Meanwhile, robust statistics such as median, mode, Inter-Quantile Range (IQR) are quite insensitive to single data addition/removal for datasets that are i.i.d. drawn from natural distributions. This means that robust statistics may have a small local sensitivity for most input datasets. 
The seminal work of \cite{dwork2009differential} introduced Propose-Test-Release framework to reduce the noise addition when the target function has an approximation that is a robust statistic\footnote{Sometimes the robust statistic itself is the target function, \rebuttal{i.e., $f_1=f_2$ in Algorithm \ref{alg:ptr-general}}.}. 
Here, we introduce a more general and useful version of PTR instantiated by Laplace and Gaussian mechanism. 
Given a target function $f_1$ (e.g., mean) and its robust variant $f_2$ (e.g., median), the PTR framework proceeds as follows: 
\textbf{(1) Propose:} a local sensitivity bound of the target function, $\tau$, is proposed; 
\textbf{(2) Test:} a safety margin $\Delta(S)$, which is the minimum amount of data points that we need to replace for $S$ to have local sensitivity larger than $\tau$, is computed. A private version (via Laplace mechanism) of the safety margin, $\widehat \Delta$, is compared with a threshold; 
\textbf{(3) Release:} if the safety margin is large enough, then the algorithm releases $f_2(S)$ via Gaussian mechanism with a smaller noise $\sigma_2$ (usually scaled with $\tau$). Otherwise, the algorithm release $f_1(S)$ with a larger noise $\sigma_1$ (usually scaled with $\globalSen_{f_1}$). The pseudocode is outlined in Algorithm \ref{alg:ptr-general}.

We make several remarks on Algorithm \ref{alg:ptr-general}. 

\begin{remark}
\textbf{(1)} The traditional version of PTR in the textbooks \cite{dwork2014algorithmic, vadhan2017complexity} simply refuse to output anything when the noisy safety margin $\widehat{\Delta}$ is small. Algorithm \ref{alg:ptr-general} is a more general version which output $f_1(S) + \N(0, \sigma_1^2 \iden_d)$ when the sensitivity test is failed. The textbook PTR can be thought of as a special case of Algorithm \ref{alg:ptr-general} where we set $\sigma_1 \rightarrow \infty$. \textbf{(2)} The mechanism used in the \textbf{Test} and \textbf{Release} step can be other mechanisms instead of Laplace and Gaussian. In this paper, we present our results for the PTR instantiated by these two mechanisms since we intend to apply it for differentially private SGD later. Our proof can be easily extended to other mechanisms as well. 
\textbf{(3)} The threshold $\log(1/(2\delta_0)) b$ in Line 3 is chosen so that $\Pr[\lap(0, b) > \log(1/(2\delta_0)) b] = \delta_0$. 
\textbf{(4)} The global sensitivity of $\Delta(\cdot)$ is 1 for any functions. This could be seen by noticing that, for any pair of adjacent $S, S'$, we have $d (S, \Tilde{S} ) \le d (S', \Tilde{S})+1$ for any dataset $\Tilde{S}$. 
\end{remark}

\paragraph{Direct DP Analysis of Propose-Test-Release.}

We show the DP guarantee for the Propose-Test-Release framework. To prove privacy, we need to find the worst pair of adjacent datasets $S$ and $S'$ that incurs the largest privacy loss. It is clear that the worst possible scenario of PTR is when $f_2(S)-f_2(S') > \tau$, while the algorithm still releases $f_2(S)$ with noise scaled with $\tau$. However, note that this worst possible scenario will only happen when \emph{both $\localSen(S)$ and $\localSen(S')$ are greater than $\tau$}. 
In this case, however, $\Delta$ for both $S$ and $S'$ are 0 since there are no data points we need to change for them to have local sensitivity larger than $\tau$, and thus there is no privacy loss from $\widehat \Delta$. Moreover, when $\Delta=0$, the probability that PTR will release $f_2(S) + \N(0, \sigma_2^2\iden_d)$ is at most $\delta_0$ \rebuttal{by construction}, which could be simply added to the $\delta$ term in $(\eps, \delta)$-DP. If one of $\localSen(S)$ and $\localSen(S')$ is smaller than $\tau$, then we know that $f_2(S)-f_2(S') \le \tau$, and the overall privacy parameters could be computed by Basic Composition Theorem \cite{dwork2014algorithmic}. We obtain the following differential privacy guarantee for PTR based on these observations. We defer more details of the proof to the Appendix \ref{appendix:proof}.

\begin{theorem}[Direct DP analysis for PTR]
Suppose $\globalSen_{f_1} = \globalSen_{f_2} = 1$ and $\sigma_1 = \sigma_2 / \tau$, then Algorithm \ref{alg:ptr-general} is $(\epslap^{(b)} + \epsgau^{(\sigma_1)}(\delta), \delta_0+\delta)$-DP, where $\epslap^{(b)} = e^{1/b}$ is the DP guarantee for Laplace mechanism, and $\eps_{\N}^{(\sigma_1)}(\delta) = \sigma_1\sqrt{2\log(1.25/\delta)}$. 
\label{thm:dp-for-ptr-maintext}
\end{theorem}

\subsection{RDP of Propose-Test-Release}

We then derive the RDP of Propose-Test-Release framework. The major differences between the proofs of $(\eps, \delta)$-DP and RDP bound is that, for $(\eps, \delta)$-DP we can move the probability of running into the bad scenario to the $\delta$ term, while for RDP we need to consider the ``average-case'' privacy loss. This is in fact an advantage of RDP, \rebuttal{which we illustrate it in the following simple example}.


\rebuttal{
\textbf{Comparison between $(\eps, \delta)$-DP and RDP Analysis: A motivating example.}\footnote{Details of derivation can be found in Appendix \ref{appendix:proof}.}
Suppose we have two mechanisms $\M_1$ and $\M_2$ who are $(\eps_1, \delta_1)$-DP and $(\eps_2, \delta_2)$-DP, respectively. 
Consider a simple PTR-like mechanism $\M$ that randomly picks one of mechanisms $\M_1$ and $\M_2$ to execute, each with probability $1-\delta_0$ and $\delta_0$, respectively\footnote{For the actual PTR, the $\delta_0$ is not fixed but depends on the input dataset.}. 
We can only claim that $\M$ is $(\max(\eps_1, \eps_2), (1-\delta_0) \delta_1 + \delta_0 \delta_2)$-DP, or if we know $\eps_2 \gg \eps_1$ we can also move the ``bad case'' probability $\delta_0$ to the $\delta$ term and obtain $(\eps_1, \delta_0 + \delta_1)$-DP. 
However, if we know the RDP guarantee of $\M_1$ and $\M_2$ as
$E_\alpha( \M_1(S) \| \M_1(S') ) \le f_\alpha(\eps_1)$
and 
$E_\alpha( \M_2(S) \| \M_2(S') ) \le f_\alpha(\eps_2)$\footnote{Recall that $f_\alpha(\eps) = \exp((\alpha-1)\eps)$ where if $\M$ is $(\alpha, \eps)$-RDP then $E_\alpha(\M(S) \| \M(S')) \le f_\alpha(\eps)$.}, 
then $E_\alpha(\M(S) \| \M(S'))$ can be simply bounded by $(1-\delta_0)f_\alpha(\eps_1) + \delta_0 f_\alpha(\eps_2)$. 
Compared with $(\eps, \delta)$-DP analysis, there are no extra inequalities used in RDP analysis of $\M$ except for the RDP guarantee of $\M_1$ and $\M_2$. 
Thus, RDP is more favorable in for PTR's privacy analysis, 
especially when $\delta_0$ is close to the target $\delta$. 
}

\begin{theorem}[RDP analysis of PTR]
Suppose $\globalSen_{f_1} = \globalSen_{f_2} = 1$ and $\sigma_1 = \sigma_2 / \tau$. 
Then for any $\alpha > 1$, Algorithm \ref{alg:ptr-general} is $(\alpha, \eps_{\ptr}(\alpha))$-RDP for 
\begin{align}
    &\eps_{\ptr} (\alpha) \le \max \left( \rdpfunc^{-1}\left( (1-\delta_0) \rdpfunc \left( \epsRgau^{(\sigma_1)}(\alpha) \right)+ \delta_0 \rdpfunc \left( \epsRgau^{(\sigma_2)}(\alpha) \right) \right),
    \epsRgau^{(\sigma_1)}(\alpha) + \epsRlap^{(b)}(\alpha) \right) \label{eq:ptr-rdp-bound} \nonumber
\end{align}
\label{thm:rdp-for-ptr-maintext}
\end{theorem}


The above result implies that, with appropriately chosen $\delta_0$, the privacy loss of PTR is the same as directly adding up the privacy loss from $\widehat \Delta$ and from releasing $f_1(S) + \N(0, \sigma_1^2\iden_d)$ regardless of the value of $\widehat \Delta$. 
Compared with naively releasing $f_1(S) + \N(0, \sigma_1^2\iden_d)$, PTR will only pay an extra cost for privatizing $\Delta$, independent from the privacy cost for the case of releasing $f_2(S) + \N(0, \sigma_2^2\iden_d)$! Thus, the algorithm could potentially add much smaller noise while introducing just a little extra privacy loss. 
In Figure \ref{fig:rdp-vs-dp-value}, we convert this RDP bound to $(\eps, \delta)$-DP via the RDP-DP conversion formula by \cite{balle2020hypothesis}, and compare with the $(\eps, \delta)$-DP bound obtained through direct analysis in Theorem \ref{thm:dp-for-ptr-maintext}. 
Although due to the loss in RDP-DP conversion, the $\eps$ parameter by RDP is worse than the one by direct analysis when $\delta_0$ is far smaller than the target $\delta$, it is tighter when $\delta_0$ is close to $\delta$.
\add{A larger $\delta_0$ means PTR has a greater chance to release $f_2(S) + \N(0, \sigma_2^2\iden_d)$, which leads to a better utility. Thus, the privacy parameter converted from the RDP bound is more favorable.}

\add{We would like to stress that in the theorem, most of the conditions on the parameters such as $\globalSen$ or $\sigma_1, \sigma_2$ can be easily relaxed. The conditions in the theorems are mainly used to make the presented bound more clean and interpretable.}



\begin{figure}[t]
    \centering
    \includegraphics[width=\columnwidth]{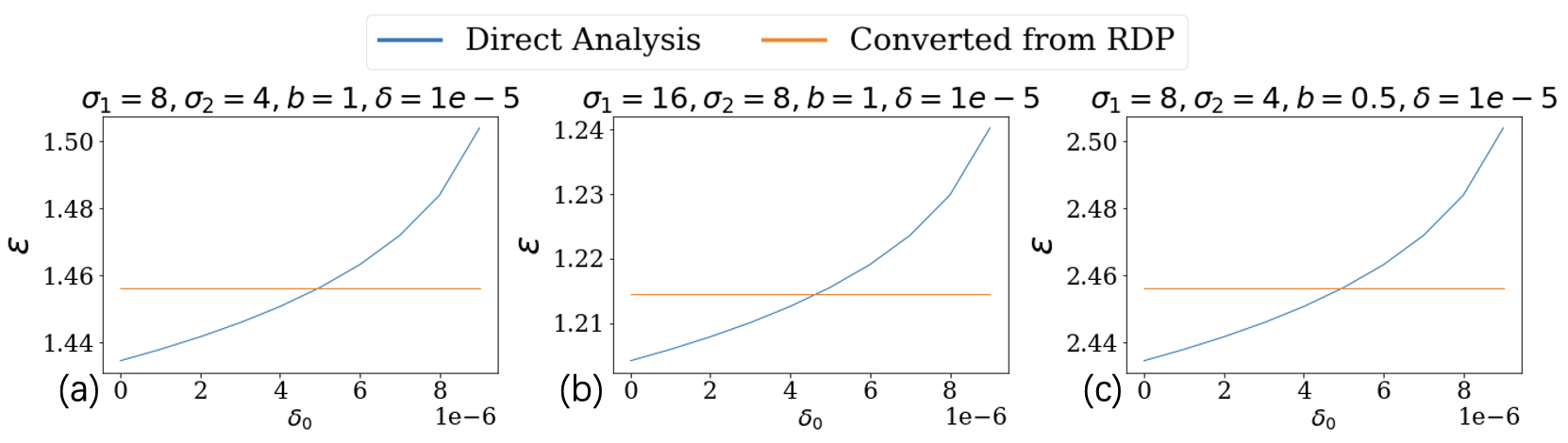}
    \caption{
    \add{
    The $\eps$ parameter of the $(\eps, \delta)$-DP guarantee of PTR when $\delta = 10^{-5}$ for different noise scales. 
    We convert the RDP bound in Theorem \ref{thm:rdp-for-ptr-maintext} to $(\eps, \delta)$-DP by the RDP-DP conversion formula from \cite{balle2020hypothesis}, and compare it with the $\eps$ obtained from the direct analysis in Theorem \ref{thm:dp-for-ptr-maintext}. 
    For the bound converted from RDP, we search for the optimal $\alpha \in [1, 200]$. 
    The bound is constant across different $\delta_0$ since when $\delta_0$ is small, the RDP for PTR will take the second term in (\ref{eq:ptr-rdp-bound}). 
    }
    }
    \label{fig:rdp-vs-dp-value}
\end{figure}

\subsection{RDP for Poisson Subsampled Propose-Test-Release}

\newcommand{\worst}{\mathcal{B}_{\text{0}}}
\newcommand{\easy}{\mathcal{B}_{\text{1}}}
\newcommand{\hard}{\mathcal{B}_{\text{2}}}

\newcommand{\rlapin}[1]{\mathrm{R}_q^{(#1)}}
\newcommand{\rlaptin}[1]{\widetilde{\mathrm{R}}_q^{(#1)}}

\add{
We further derive the RDP for Poisson subsampled PTR. Here, ``subsampled'' means the input dataset will be subsampled first before feeding to PTR. Poisson subsampling means each data point $x$ will be included with probability $q$ independently. One way of obtaining the RDP for any subsampled mechanism is by simply plugging in the RDP bound of the original algorithm into the privacy amplification formula for general mechanisms (e.g., \cite{wang2019subsampled, zhu19poisson}). Here, we directly derive the RDP for Poisson subsampled PTR in a white-box manner. 
Denote $\M_0$ as the random variable for PTR's output on dataset $S$, and $\M_1$ for PTR's output on dataset $S' = S \cup \{x\}$, and $\M = (1-q)\M_0 + q\M_1$. Compared with the analysis for subsampled Gaussian mechanism by \cite{mironov2019r}, there are two main difficulties for extending the arguments of quasi-convexity of \Renyi divergence for PTR: \textbf{(1)} the distribution of PTR's output may not be centrally symmetric, thus we need to bound both $D_\alpha( \M || \M_0 )$ and $D_\alpha( \M_0 || \M )$. \textbf{(2)} the conditional distribution $\M|\widehat \Delta$ cannot be decomposed as $(1-q) \M_0|\widehat \Delta + q \M_1|\widehat \Delta$. A big part of our novelty in the proof is about addressing these two challenges. We defer details of the proof to Appendix \ref{appendix:proof}. 
}

\begin{theorem}[RDP analysis of sub-sampled PTR (the \emph{``white-box bound''})]
\label{thm:subsampled-ptr}
Let $q$ be the subsampling probability. 
Suppose $\globalSen_{f_1} = \globalSen_{f_2} = 1$ and $\sigma_1 = \sigma_2 / \tau$. 
When $q, \sigma_1, \sigma_2$, and $\alpha$ satisfy certain conditions, 
we have 
\begin{align}
    \eps_{\ptr \circ \poisson}(\alpha) 
    \le \rdpfunc^{-1} \left( \max( \worst, \easy, \hard ) \right) \nonumber
\end{align}
where 
$\worst = 1 + 2q^2 \alpha (\alpha-1) ( \frac{1-\delta_0}{\sigma_1^2} + \frac{\delta_0}{\sigma_2^2} )$, 
$\easy = \rlapin{\alpha} + 2\alpha (\alpha-1) [ \rlapin{\alpha} - 2(1-q)\rlapin{\alpha-1} + (1-q)^2 \rlapin{\alpha-2} ]$, 
$
\hard = \rlaptin{\alpha} + 2\alpha (\alpha-1) [ \rlaptin{\alpha} - 2(1-q)\rlaptin{\alpha+1} + (1-q)^2 \rlaptin{\alpha+2} ]
$, 
with $\rlapin{\alpha} = \E_{s \sim \mu_0} \left[\left(\frac{\mu(s)}{\mu_0(s)}\right)^\alpha \right]$
and 
$\rlaptin{\alpha} = \E_{s \sim \mu} \left[\left(\frac{\mu_0(s)}{\mu(s)}\right)^\alpha \right]$
for $\mu_0 \sim \lap(0, b)$ and $\mu \sim (1-q) \lap(0, b) + q\lap(1, b)$. 
\end{theorem}

\begin{remark}
Both $\rlapin{\alpha}$ and $\rlaptin{\alpha}$ can be computed easily since they either have closed form or can be accurately computed via numerical integration in bounded range. 
\end{remark}

This bound may not be very interpretable, which is a typical feature for the privacy amplification bounds as they are meant to be implemented in practice. After all, constant matters for differential privacy practitioners! 
To show the tightness of our bound, we plug in the RDP of original PTR (Theorem \ref{thm:rdp-for-ptr-maintext}) to the current tightest privacy amplification formula for general mechanisms derived in \cite{zhu19poisson} (the \emph{``black-box bound''}), and compare it with our white-box bound in Theorem \ref{thm:subsampled-ptr}. 
As we can see from Figure \ref{fig:rdp-compare}, Theorem \ref{thm:subsampled-ptr} (orange curve) is much tighter than the black-box bound (blue curve). 
Moreover, it is very close to the lower bound of privacy amplification by \cite{zhu19poisson} (green line), which means that Theorem \ref{thm:subsampled-ptr} is near optimal. 
\add{
In Figure \ref{fig:subsampled-comp}, we illustrate the application of subsampled RDP bound in Moment Accountant. We can see that the privacy parameters for the composed mechanism obtained based on our white-box bound (Theorem \ref{thm:subsampled-ptr}) is tighter than the one by black-box bound, as well as the one by directly composing Theorem \ref{thm:dp-for-ptr-maintext} with strong composition theorem \cite{kairouz2015composition}. }
\add{
We remark that Direct Analysis achieves lower privacy loss with very few iterations since we set $\delta_0 = 10^{-8}$ to allow more iterations for Strong Composition of $(\eps, \delta)$-DP, which leads to a better $\eps$ for a single iteration (see Figure \ref{fig:rdp-vs-dp-value}). }

An exciting byproduct during our research on subsampled PTR is the privacy amplification bound for a $\M$ that outputs $\M_1(S), \M_2(S)$ sequentially and \emph{independently}. This result is important since such a $\M$ is instantiated in many existing techniques in improving differentially private deep learning. 
We defer this result to the Appendix.

\begin{figure}[t]
    \centering
    \includegraphics[width=\columnwidth]{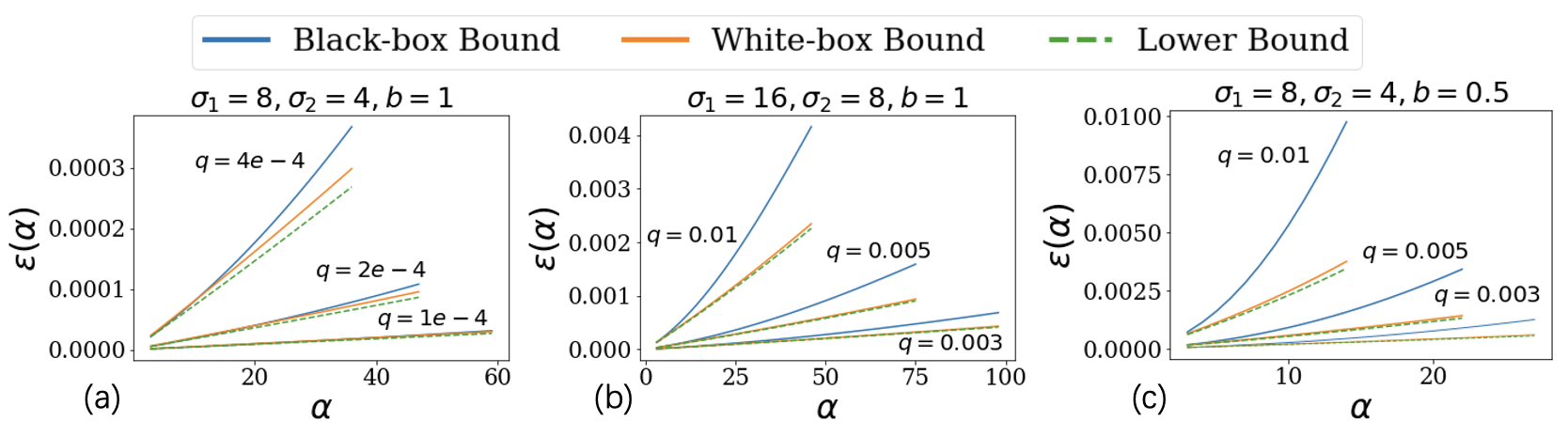}
    \caption{\add{The RDP parameter $\eps(\alpha)$ of the subsampled PTR as a function of order $\alpha$. We compare the white-box bound from Theorem \ref{thm:subsampled-ptr}, and the black-box bound as well as the lower bound by plugging in the RDP of PTR in Theorem \ref{thm:rdp-for-ptr-maintext} to the privacy amplification upper/lower bounds for general mechanisms from \cite{zhu19poisson}. 
    }
    }
    \label{fig:rdp-compare}
\end{figure}

\begin{figure}[t]
    \centering
    \includegraphics[width=\columnwidth]{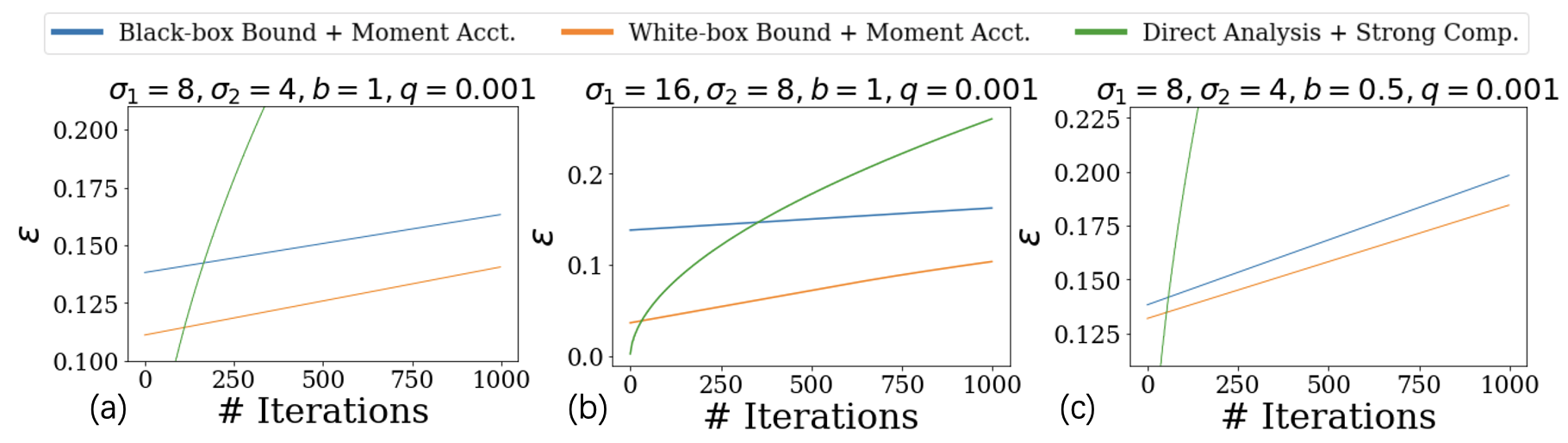}
    \caption{\add{
    Illustration of the use of our Theorem \ref{thm:subsampled-ptr} in moments accountant. We plot the the privacy loss $\eps$ for $\delta=10^{-5}$ after different rounds of composition. We set $\delta_0 = 10^{-8}$ here to allow more iterations for Strong Composition of $(\eps, \delta)$-DP. }
    }
    \label{fig:subsampled-comp}
\end{figure}

\section{Differentially Private and Robust SGD with Propose-Test-Release}
\label{sec:robust-and-dp-sgd}

In this section, we demonstrate the application of the RDP for (subsampled) PTR algorithm in privatizing robust SGD algorithms. 



\newcommand{\mean}{\mathrm{SUM}}
\newcommand{\tmean}{\mathrm{TSUM}}
\newcommand{\nfail}{F}

\rebuttal{
\textbf{Attack Model. }
Mini-batch SGD is a common method for training deep neural networks. Despite its strong convergence properties in the standard settings, it is well known that even a small fraction of corrupted gradients can lead SGD to an arbitrarily poor solution \cite{ben2000robust, bertsimas2011theory}. 
An important attack model is called Byzantine contamination framework or \emph{Byzantine failure} \cite{lamport1982byzantine}. 
Consider an optimization problem with $n$ stochastic gradient oracles; at each iteration there are up to $\nfail$ gradient oracles are corrupted (usually referred as \emph{Byzantine agents} in the context of distributed SGD). The identity of corrupted oracles is a priori unknown. 
As the corrupted gradients can be arbitrarily skewed, this attack model is able to capture many important 
scenarios including \textbf{corruption in feature} (e.g., existence of outliers), \textbf{corruption in gradients} (e.g., hardware failure, unreliable communication channels during distributed training) and \textbf{corruption in labels} (e.g. label flip (backdoor) attacks). 
}

\rebuttal{
A popular class of defense strategies against Byzantine failure is to use robust gradient aggregation rules such as trimmed-mean \citep{gupta2021byzantine}, coordinate-wise median \citep{yin2018byzantine}, and geometric median \citep{acharya2021robust}.}
Since both privacy and robustness are essential for ML applications, it is important to develop learning algorithms that achieve differential privacy and robustness simultaneously. 

\textbf{Baseline: privatizing robust SGD with Gaussian Mechanism.}
A natural (but naive) way to accomplish such a goal for SGD-based algorithms is to privatize each iteration of robust SGD through Gaussian mechanism. Specifically, in each iteration, we add Gaussian noise to the aggregated gradient, where noise magnitude scale with the global sensitivity of the robust aggregation function. The privacy parameter for the entire process is computed through Moment Accountant. The global sensitivity can be obtained by either upper bounding the Lipschitz constant of the objective function, or clipping each gradient according to a certain threshold. However, this method could potentially add over-conservative noise to the aggregated gradients, as the local sensitivity of robust aggregation functions is usually very low for most input gradient sets. 

\textbf{Privatizing robust SGD with PTR.}
We improve the naive method with the PTR framework, where we instantiate it by showing how to privatize trimmed-mean SGD method from \cite{gupta2021byzantine}. 
We first define trimmed-sum\footnote{We use trimmed-sum instead of trimmed-mean for the ease of analysis, as trimmed-mean essentially just scales the learning rate.} aggregation function, and show that the $\Delta$ for trimmed-sum could be computed efficiently. 
Given a dataset $S = \{x_1, \ldots, x_m\}$, let $x_{(k)}$ denote the $k$th smallest data point among $S$ in $\ell_2$ norm, i.e., $\norm{x_{(1)}} \le \norm{x_{(2)}} \le \ldots \le \norm{x_{(m)}}$. The \emph{$\nfail$-trimmed sum of $S$} is defined as $\tmean_\nfail(S) = \sum_{i=1}^{m-\nfail} x_{(i)}$ if $m > \nfail$, or simply 0 if $m \le \nfail$. Define $\localSen_{f}^{(r)}(S) = \sup_{S, \Tilde{S}: d(S, \Tilde{S})=r} \localSen_f(\Tilde{S})$, which is the maximum local sensitivity that can be achieved by adding/removing $r$ elements from $S$. Thus, $\localSen_{f}^{(0)}(S) = \localSen_{f}(S)$, and the safety margin $\Delta(S) = \min \{r: \localSen_{f}^{(r)}(S) > \tau\}$. Therefore, as long as we can efficiently compute $\localSen_{f}^{(r)}(S)$, we can efficiently compute $\Delta(S)$ in linear time by simply enumerate all $r$ from 0 to $m$ and terminate once $\localSen_{f}^{(r)}(S) > \tau$. In fact, $\localSen_{\tmean_{\nfail}}^{(r)}(S)$ can be computed in $O(1)$ time. 
\begin{theorem}
$\localSen_{\tmean_{\nfail}}^{(r)}(S) = \norm{x_{(m-F+1+r)}}$ if $r \le \nfail-1$, or $\globalSen_{\tmean_{\nfail}}$ if $r > \nfail-1$. 
\label{thm:localsenr}
\end{theorem}

Back to the goal of designing differentially private and robust SGD. Fix a positive integer $\nfail$ which is a potential upper bound for the number of corrupted gradients. 
We instantiate a PTR-based gradient aggregation algorithm by instantiating $f_1(\cdot)$ as $\mean(S) = \sum_{x \in S}x$, and $f_2(\cdot)$ as $\tmean_{\nfail}(\cdot)$. 
We simply plugging in this PTR algorithm as the gradient aggregation function for regular SGD. 
We show the convergence guarantee of this algorithm in the presence of at most $\nfail$ gradients being corrupted for the case that the loss function is Lipschitz. 
\begin{theorem}[Informal]
\label{thm:convergence}
When the loss function is Lipschitz in model parameters, SGD with PTR instantiated by $\mean$ and $\tmean$ provides convergence guarantee for update to $\nfail$ gradients being corrupted arbitrarily for any $\nfail<m/2$. 
\end{theorem}

We make two remarks of our specific design choices for implementing PTR-based SGD in practice.   
\begin{remark}
\textbf{(1)}
In practice, $\nfail$ could be adjusted dynamically during the model training, based on the value of $\widehat \Delta$\footnote{\rebuttal{This does not introduce extra privacy leakage as DP is closed under post-processing.}}. That is, if $\widehat \Delta \le B$, it means that the current gradient batch has many outliers, thus we increase $\nfail$ for the next iteration; if $\widehat \Delta > B$, it means that gradients are relatively concentrated and few of them are corrupted, thus we can decrease $\nfail$ afterwards. 
\textbf{(2)} We set $f_1$ as $\mean$ instead of $\tmean$ since $\nfail$ is usually set to be larger than the actual corrupted gradients; in this case, $\widehat \Delta \le B$ means that there are many benign gradients also have large norms. These benign gradients usually correspond to the data points that are misclassified by the partially trained models, which is important for improving model performance. 
The pseudocode for the algorithm is deferred to the Appendix \ref{appendix:proof}.
\end{remark}

\begin{table}[t]
\centering
\resizebox{\columnwidth}{!}{
\begin{tabular}{@{}ccccccc@{}}
\toprule
\multirow{2}{*}{\textbf{Dataset}} & \multirow{2}{*}{\textbf{\begin{tabular}[c]{@{}c@{}}Corruption\\ Type\end{tabular}}} & \multirow{2}{*}{\textbf{CR}} & \multicolumn{2}{c}{\textbf{$\eps = 3.0$}}                       & \multicolumn{2}{c}{\textbf{$\eps=5.0$}}                         \\
                                  &                                                                                     &                              & \textbf{TSGD+Gaussian} & \textbf{TSGD+PTR}                      & \textbf{TSGD+Gaussian} & \textbf{TSGD+PTR}                      \\ \midrule
\multicolumn{1}{l}{}              & \multicolumn{1}{l}{}                                                                & \textbf{0}                   & 87.53\%                & 91.43\% $\textcolor{red}{(+3.9\%)}$    & 81.31\%                & 92.8\% $\textcolor{red}{(+11.49\%)}$   \\
\textbf{MNIST}                    & \textbf{Label}                                                                      & \textbf{0.1}                 & 90.15\%                & 93.29\% $\textcolor{red}{(+3.138\%)}$  & 89.15\%                & 94.608\% $\textcolor{red}{(+5.456\%)}$ \\
\textbf{}                         & \textbf{}                                                                           & \textbf{0.2}                 & 91.33\%                & 92.704\% $\textcolor{red}{(+1.374\%)}$ & 92.47\%                & 94.084\% $\textcolor{red}{(+1.614\%)}$ \\
\textbf{}                         & \multirow{2}{*}{\textbf{Feature}}                                                   & \textbf{0.1}                 & 89.47\%                & 92.28\% $\textcolor{red}{(+2.812\%)}$  & 86.53\%                & 93.388\% $\textcolor{red}{(+6.854\%)}$ \\
\textbf{}                         &                                                                                     & \textbf{0.2}                 & 91.71\%                & 92.29\% $\textcolor{red}{(+0.582\%)}$  & 86.99\%                & 92.864\% $\textcolor{red}{(+5.872\%)}$ \\
\textbf{}                         & \multirow{2}{*}{\textbf{Gradient}}                                                  & \textbf{0.1}                 & 90.27\%                & 91.7\% $\textcolor{red}{(+1.43\%)}$    & 88.41\%                & 91.68\% $\textcolor{red}{(+3.27\%)}$   \\
                                  &                                                                                     & \textbf{0.2}                 & 91.13\%                & 91.45\% $\textcolor{red}{(+0.32\%)}$   & 88.70\%                & 90.31\% $\textcolor{red}{(+1.61\%)}$   \\ \midrule
\multirow{2}{*}{}                 & \multirow{2}{*}{}                                                                   & \multirow{2}{*}{}            & \multicolumn{2}{c}{\textbf{$\eps = 7.0$}}                       & \multicolumn{2}{c}{\textbf{$\eps=10.0$}}                        \\
                                  &                                                                                     &                              & \textbf{TSGD+Gaussian} & \textbf{TSGD+PTR}                      & \textbf{TSGD+Gaussian} & \textbf{TSGD+PTR}                      \\ \midrule
\multicolumn{1}{l}{}              & \multicolumn{1}{l}{}                                                                & \textbf{0}                   & 57.94\%                & 58.92\% $\textcolor{red}{(+0.98\%)}$   & 58.82\%                & 60.13\% $\textcolor{red}{(+1.31\%)}$   \\
\textbf{CIFAR10}                  & \textbf{Label}                                                                      & \textbf{0.1}                 & 55.48\%                & 56.262\% $\textcolor{red}{(+0.78\%)}$  & 57.42\%                & 58.218\% $\textcolor{red}{(+0.8\%)}$   \\
\textbf{}                         & \textbf{}                                                                           & \textbf{0.2}                 & 49.69\%                & 50.946\% $\textcolor{red}{(+1.26\%)}$  & 52.95\%                & 53.106\% $\textcolor{red}{(+0.152\%)}$ \\
\textbf{}                         & \multirow{2}{*}{\textbf{Feature}}                                                   & \textbf{0.1}                 & 55.70\%                & 56.628\% $\textcolor{red}{(+0.93\%)}$  & 57.36\%                & 58.678\% $\textcolor{red}{(+1.318\%)}$ \\
\textbf{}                         &                                                                                     & \textbf{0.2}                 & 55.13\%                & 55.526\% $\textcolor{red}{(+0.398\%)}$ & 56.93\%                & 57.336\% $\textcolor{red}{(+0.404\%)}$ \\
\textbf{}                         & \multirow{2}{*}{\textbf{Gradient}}                                                  & \textbf{0.1}                 & 55.81\%                & 57.29\% $\textcolor{red}{(+1.48\%)}$   & 57.21\%                & 58.96\% $\textcolor{red}{(+1.75\%)}$   \\
                                  &                                                                                     & \textbf{0.2}                 & 54.45\%                & 56.31\% $\textcolor{red}{(+1.86\%)}$   & 56.11\%                & 57.48\% $\textcolor{red}{(+1.37\%)}$   \\ \midrule
\multirow{2}{*}{}                 & \multirow{2}{*}{}                                                                   & \multirow{2}{*}{}            & \multicolumn{2}{c}{\textbf{$\eps = 2.0$}}                       & \multicolumn{2}{c}{\textbf{$\eps=2.5$}}                         \\
                                  &                                                                                     &                              & \textbf{TSGD+Gaussian} & \textbf{TSGD+PTR}                      & \textbf{TSGD+Gaussian} & \textbf{TSGD+PTR}                      \\ \midrule
\multicolumn{1}{l}{}              & \multicolumn{1}{l}{}                                                                & \textbf{0}                   & 72.94\%                & 76.44\% $\textcolor{red}{(+3.5\%)}$    & 79.15\%                & 81.02\% $\textcolor{red}{(+1.87\%)}$   \\
\textbf{EMNIST}                   & \textbf{Label}                                                                      & \textbf{0.1}                 & 72.60\%                & 75.66\% $\textcolor{red}{(+3.06\%)}$   & 79.38\%                & 80.63\% $\textcolor{red}{(+1.25\%)}$   \\
\textbf{}                         & \textbf{}                                                                           & \textbf{0.2}                 & 70.03\%                & 72.62\% $\textcolor{red}{(+2.59\%)}$   & 77.48\%                & 79.19\% $\textcolor{red}{(+1.71\%)}$   \\
\textbf{}                         & \multirow{2}{*}{\textbf{Feature}}                                                   & \textbf{0.1}                 & 69.60\%                & 74.01\% $\textcolor{red}{(+4.41\%)}$   & 77.80\%                & 81.04\% $\textcolor{red}{(+3.24\%)}$   \\
\textbf{}                         &                                                                                     & \textbf{0.2}                 & 70.04\%                & 74.76\% $\textcolor{red}{(+4.72\%)}$   & 78.99\%                & 80.74\% $\textcolor{red}{(+1.75\%)}$   \\
\textbf{}                         & \multirow{2}{*}{\textbf{Gradient}}                                                  & \textbf{0.1}                 & 71.75\%                & 76.19\% $\textcolor{red}{(+4.44\%)}$   & 77.32\%                & 77.73\% $\textcolor{red}{(+0.41\%)}$   \\
                                  &                                                                                     & \textbf{0.2}                 & 70.76\%                & 74.65\% $\textcolor{red}{(+3.89\%)}$   & 76.68\%                & 77.17\% $\textcolor{red}{(+0.49\%)}$   \\ \bottomrule
\end{tabular}
}
\caption{Model Accuracy under different privacy budgets and corruption settings. Every statistic is averaged over 5 runs with different random seed. 
The improvement of $\tsgd+\ptr$ over $\tsgd+\gau$ is highlighted in the red text. $\eps$s are chosen differently for different datasets since the best accuracy-privacy tradeoff point varied for datasets.
\rebuttal{`CR' means corruption ratio.}
} 
\label{tb:full}
\end{table}

\subsection{Evaluation}
\label{sec:eval}

In this section, we empirically evaluate the performance of PTR-based private and robust SGD. 
Specifically, we compare the utility-privacy tradeoff between the algorithm that privatize trimmed mean-based SGD with PTR ($\tsgd+\ptr$), and the baseline algorithm that privatize trimmed mean-based SGD with Gaussian mechanism ($\tsgd+\gau$). 

\textbf{Experiment Settings. }
We evaluate the utility-privacy tradeoff of the two robust and differentially private SGD algorithms under three common types of attack: label, feature, and communicated gradient corruption. 
For label corruption, we randomly flip of label of certain amount of data points. 
For feature corruption, we add Gaussian noise from $\N(0, 100)$ directly to the corrupted images. 
For gradient corruption, we add Gaussian noise from $\N(0, 100)$ to the true gradients. 
We experiment on three classic datasets, MNIST, CIFAR10, and EMNIST. We vary different corruption ratios and privacy parameter $\eps$. Experiment details are deferred to Appendix.

\textbf{Results. }
The comparison of model test accuracy under given corruption settings and privacy parameters are summarized in Table \ref{tb:full}. We highlight the improvement of $\tsgd+\ptr$ over the baseline $\tsgd+\gau$ in red texts. As we can see, for all settings, $\tsgd+\ptr$ outperforms and often works significantly better than $\tsgd+\gau$. This demonstrates that, while $\tsgd+\ptr$ may introduce extra privacy loss in the \textbf{test} step, the performance gain from adding smaller noise in the \textbf{release} step overshadows it. 

We also observe two interesting phenomena in the experiment: \textbf{(1)} a higher corruption ratio may not necessarily lead to worse model performance for trimmed mean-based robust SGD, especially for MNIST dataset. This is because the high norm gradients can be either corrupted, or benign but the partially-trained model misclassifies the corresponding data points. The latter case is extremely important for improving model performance compared with the gradients of data points are already being classified correctly. When the corruption ratio is small, more benign gradients are being trimmed, which may lead to worse model performance. \textbf{(2)} For $\tsgd+\gau$, a larger privacy budget may not lead to better model performance on MNIST dataset. This is because when the training accuracy reaches the peak, there are many trimmed gradients whose corresponding training data points are already correctly classified; continuing training without those data points and with large noise may result in catastrophic forgetting.

\section{Conclusion and Limitation}
\label{sec:conclusion}

This work derives the \Renyi Differential Privacy for propose-test-release framework as well as its subsampled version. 
With the RDP bound for the PTR framework, this work demonstrate the application of PTR in training differentially private and robust models. \add{
One limitation of PTR is that it does not work well in privatizing coordinate-wise median in high-dimensional space. The global sensitivity of coordinate-wise median is far greater than the one for mean, which results in huge privacy loss.} 

\section*{Acknowledgement}
We are grateful to anonymous reviewers at NeurIPS for
valuable feedback. This work was supported in part by
the National Science Foundation under grants CNS2131910, CNS-1553437, CNS-1553301, CNS-1704105,
and CNS-1953786, the ARL’s Army Artificial Intelligence Innovation Institute (A2I2), the Office of Naval Research Young Investigator Award, the Army Research Office Young Investigator Prize, Schmidt DataX award, Princeton E-ffiliates Award, Amazon-Virginia Tech Initiative in Efficient and Robust Machine Learning, and Princeton's Gordon Y. S. Wu Fellowship.

\newpage

\bibliography{ref.bib}
\bibliographystyle{alpha}

\newpage

\appendix

\section{Proofs, Additional Theoretical Results \& Discussion}
\label{appendix:proof}

\subsection{$(\eps, \delta)$-DP and R\'enyi DP for Propose-Test-Release}

We consider two adjacent datasets $S, S'$ where $S' = S \cup \{x\}$. We denote the threshold $B = \log(1/(2\delta_0))b$. Note that we have $\Pr[\lap(0, b) > B] = \delta_0$. The output of Algorithm \ref{alg:ptr-general} on dataset $S$ is a sample from a joint distribution $(\widehat \Delta, \M)(S)$ where $\widehat \Delta(S) = \lap(\Delta(S), b)$ and $\M(S)|_{\widehat \Delta} = \N(f_1(S), \sigma_1^2) \Ind[\widehat \Delta \le B] + \N(f_2(S), \sigma_2^2) \Ind[\widehat \Delta > B]$. 

\begin{customthm}{\ref{thm:dp-for-ptr-maintext}}[restated]
Suppose $\globalSen_{f_1} = \globalSen_{f_2} = 1$ and $\sigma_1 = \sigma_2 / \tau$, then Algorithm \ref{alg:ptr-general} is $(\epslap^{(b)} + \epsgau^{(\sigma_1)}(\delta), \delta_0+\delta)$-DP.  
\end{customthm}
\begin{proof}
We consider two cases for $\localSen_{f_2}(S)$ and $\localSen_{f_2}(S')$.

\textbf{Case 1: both $\localSen_{f_2}(S)$ and $\localSen_{f_2}(S')$ are greater than $\tau$. }
In this case, we have $\Delta(S) = \Delta(S') = 0$ (recall that $\Delta$ refers to the minimum amount of data addition/removal to make the local sensitivity $> \tau$). Therefore, there are no privacy loss in $\widehat \Delta$. Besides, the probability that PTR releases $f_2(S) + \N(0, \sigma_2^2)$ is at most $\Pr[\widehat\Delta > B] = \Pr[\lap(0, b) > B] = \delta_0$. Therefore, with probability at least $1-\delta_0$, the PTR is $(\eps_\N^{(\sigma_1)}(\delta), \delta)$-DP, and overall it is $(\eps_\N^{(\sigma_1)}(\delta), \delta+\delta_0)$-DP. 

\textbf{Case 2: at least one of $\localSen_{f_2}(S)$ and $\localSen_{f_2}(S')$ are smaller than $\tau$.}
In this case, we know that $\Pr[ \M(S) \in T ] \le e^{\eps_\N^{(\sigma_1)}(\delta)}\Pr[ \M(S') \in T ] + \delta$ regardless of the value of $\widehat \Delta$. Thus, by basic composition theorem, PTR in this case is $\left(\eps_\lap^{(b)} + \eps_\N^{(\sigma_1)}(\delta), \delta\right)$-DP. 

Therefore, PTR is $(\epslap^{(b)} + \epsgau^{(\sigma_1)}(\delta), \delta_0+\delta)$-DP overall. 
\end{proof}


\rebuttal{
\paragraph{Comparison between $(\eps, \delta)$-DP and RDP Analysis: A motivating example (expanded).}
Suppose we have two mechanisms $\M_1$ and $\M_2$ who are $(\eps_1, \delta_1)$-DP and $(\eps_2, \delta_2)$-DP, respectively. 
Consider a simple PTR-like mechanism $\M$ that randomly picks one of mechanisms $\M_1$ and $\M_2$ to run, each with probability $1-\delta_0$ and $\delta_0$\footnote{For the actual PTR, the $\delta_0$ is not fixed but depends on the input dataset.}. 
A straightforward $(\eps, \delta)$-DP analysis for $\M$ can be given as follows: for any possible event $T$,
\begin{align}
    \Pr[\M(S) \in T] 
    &= (1-\delta_0) \Pr[\M_1(S) \in T] + \delta_0 \Pr[\M_2(S) \in T] \\
    &\le (1-\delta_0) [ e^{\eps_1} \Pr[\M_1(S') \in T] + \delta_1] + \delta_0 [ e^{\eps_2} \Pr[\M_2(S') \in T] + \delta_2] \\
    &= e^{\eps_1}(1-\delta_0)\Pr[\M_1(S') \in T] + e^{\eps_2} \delta_0 \Pr[\M_2(S') \in T] + (1-\delta_0) \delta_1 + \delta_0 \delta_2 \\
    &\le e^{\max(\eps_1, \eps_2)} \Pr[\M(S) \in T] + (1-\delta_0) \delta_1 + \delta_0 \delta_2
\end{align}
That is, $\M$ is $(\max(\eps_1, \eps_2), (1-\delta_0) \delta_1 + \delta_0 \delta_2)$-DP. Without further information, this bound is the best we can do since it is tight when there exists event $T$ such that $\Pr[\M_1(S') \in T] = 0$ while $\Pr[\M_2(S') \in T] > 0$. 
Alternatively, if we know $\eps_2 \gg \eps_1$ we can also move the probability $\delta_0$ to the $\delta$ term and obtain $(\eps_1, \delta_0 + \delta_1)$ (which is the case for Theorem \ref{thm:dp-for-ptr-maintext}). 

However, if we know the RDP guarantee of $\M_1$ and $\M_2$ as
$E_\alpha( \M_1(S) \| \M_1(S') ) \le f_\alpha(\eps_1)$
and 
$E_\alpha( \M_2(S) \| \M_2(S') ) \le f_\alpha(\eps_2)$\footnote{Recall that $f_\alpha(\eps) = \exp((\alpha-1)\eps)$ where if $\M$ is $(\alpha, \eps)$-RDP then $E_\alpha(\M(S) \| \M(S')) \le f_\alpha(\eps)$.}, 
then $E_\alpha(\M(S) \| \M(S'))$ can be simply bounded as
\begin{align}
\E_{\M(S)} \left[ \left( \frac{\mu_{\M(S')}}{\mu_{\M(S)}} \right)^\alpha \right]
&= (1-\delta_0) \E_{\M_1(S)} \left[ \left( \frac{\mu_{\M_1(S')}}{\mu_{\M_1(S)}} \right)^\alpha \right] + \delta_0 \E_{\M_2(S)} \left[ \left( \frac{\mu_{\M_2(S')}}{\mu_{\M_2(S)}} \right)^\alpha \right] \\
&\le (1-\delta_0)f_\alpha(\eps_1) + \delta_0 f_\alpha(\eps_2)
\end{align}
Compared with $(\eps, \delta)$-DP analysis, there are no extra inequalities used in RDP analysis of $\M$ except for the RDP guarantee of $\M_1$ and $\M_2$. 
Thus, RDP is more favorable in for PTR's privacy analysis, 
especially when $\delta_0$ is close to the target $\delta$. 
}

\begin{customthm}{\ref{thm:rdp-for-ptr-maintext}}[restated]
Suppose $\globalSen_{f_1} = \globalSen_{f_2} = 1$ and $\sigma_1 = \sigma_2 / \tau$. Then for any $\alpha > 1$, Algorithm \ref{alg:ptr-general} is $(\alpha, \eps_{\ptr}(\alpha))$-RDP for 
\begin{align}
    \eps_{\ptr} (\alpha) \le \max \left( \rdpfunc^{-1}\left( (1-\delta_0) \rdpfunc \left( \epsRgau^{(\sigma_1)}(\alpha) \right)+ \delta_0 \rdpfunc \left( \epsRgau^{(\sigma_2)}(\alpha) \right) \right),
    \epsRgau^{(\sigma_1)}(\alpha) + \epsRlap^{(b)}(\alpha) \right) \nonumber
\end{align}
\end{customthm}
\full{
\begin{remark}
This result implies that, the privacy loss of PTR could be as small as the direct composition of Laplace and Gaussian mechanism. We can find the optimal $\delta_0$ by computing the $\delta_0$ s.t. 
\begin{align}
    \frac{1}{\alpha-1} \log \left( (1-\delta_0) \rdpfunc \left( \epsRgau^{(\sigma_1)}(\alpha) \right)+ \delta_0 \rdpfunc \left( \epsRgau^{(\sigma_2)}(\alpha) \right) \right)
    = \epsRgau^{(\sigma_1)}(\alpha) + \epsRlap^{(b)}(\alpha)
\end{align}
which leads to 
\begin{align}
    \delta_0 = 
    \frac{ \rdpfunc \left( \epsRgau^{(\sigma_1)}(\alpha) \right) \left(\rdpfunc \left( \epsRlap^{(b)}(\alpha) \right) - 1\right)}{\rdpfunc \left( \epsRgau^{(\sigma_2)}(\alpha)\right)-\rdpfunc \left( \epsRgau^{(\sigma_1)}(\alpha)\right)}
\end{align}
\end{remark}
}
\begin{proof}
We will denote the density of $(\widehat \Delta, \M)(S)$ as $\mu$ and that of $(\widehat \Delta, \M)(S')$ as $\mu'$. We will use $\mu(s, t)$ to denote the joint density on the pair of outputs $(s, t)$, where $s \sim \widehat \Delta(S)$ and $t \sim \M(S)|_{\widehat \Delta}$. Furthermore, when we write $\mu(s)$ it refers to the marginal density of $\mu$ on $s$, and $\mu(t|s)$ refers to the conditional density on $t$ given $s$. 

In order to bound RDP of PTR with order $\alpha$, it suffices to bound the moments $\E_{(s, t) \sim \mu} \left[ \left(\frac{\mu'(s, t)}{\mu(s, t)}\right)^\alpha \right]$ and $\E_{(s, t) \sim \mu'} \left[\left(\frac{\mu(s, t)}{\mu'(s, t)}\right)^\alpha \right]$ then take the bigger of the two bounds. 
For readability, we may abbreviate the two quantities as $\E_{\mu} \left[\left(\frac{\mu'}{\mu}\right)^\alpha \right]$ and $\E_{\mu'} \left[\left(\frac{\mu}{\mu'}\right)^\alpha \right]$. 
We do the following to decompose $\E_{(s, t) \sim \mu} \left[\left(\frac{\mu'(s, t)}{\mu(s, t)}\right)^\alpha \right]$:
\begin{align}
    &\E_{(s, t) \sim \mu} \left[\left(\frac{\mu'(s, t)}{\mu(s, t)}\right)^\alpha \right] \\
    &= \E_{(s, t) \sim \mu} \left[ \left(\frac{\mu'(s)\mu'(t|s)}{\mu(s)\mu(t|s)}\right)^\alpha \right] \\
    &= \E_{s \sim \mu} \left[
    \left( \frac{\mu'(s)}{\mu(s)} \right)^\alpha
    \E_{t \sim \mu|s} \left[
    \left(\frac{\mu'(t|s)}{\mu(t|s)}\right)^\alpha \right] \right] \\
    &= \E_{s \sim \mu} \left[
    \left( \frac{\mu'(s)}{\mu(s)} \right)^\alpha
    \left(
    \E_{t \sim \mu|s} \left[
    \left(\frac{\mu'(t|s)}{\mu(t|s)}\right)^\alpha  \right]\Ind[s \le B] +  \E_{t \sim \mu|s} \left[
    \left(\frac{\mu'(t|s)}{\mu(t|s)}\right)^\alpha \right]\Ind[s > B]
    \right)
    \right] \\
\end{align}

When $s \le B$, we know that 
\begin{align}
    \E_{t \sim \mu|s} \left[
    \left(\frac{\mu'(t|s)}{\mu(t|s)}\right)^\alpha  \right]
    &= 
    \E_{t \sim \N(f_1(S), \sigma_1^2\iden_d)} \left[ 
    \left(
    \frac{\N(t; f_1(S'), \sigma_1^2\iden_d)}{\N(t; f_1(S), \sigma_1^2\iden_d)}
    \right)^\alpha
    \right] \\
    &= \E_{t \sim \N(\bm{0}, \sigma_1^2\iden_d)} \left[ 
    \left(
    \frac{\N(t; f_1(S') - f_1(S), \sigma_1^2\iden_d)}{\N(t; \bm{0}, \sigma_1^2\iden_d)}
    \right)^\alpha
    \right] \label{eq:translation-inv} \\
    &= \E_{t \sim \N(0, \sigma_1^2)} \left[ 
    \left(
    \frac{\N(t; \norm{f_1(S') - f_1(S)}, \sigma_1^2)}{\N(t; 0, \sigma_1^2)}
    \right)^\alpha
    \right] \label{eq:rotation} \\
    &\le \E_{t \sim \N(0, \sigma_1^2 )} \left[ 
    \left(
    \frac{\N(t; 1, \sigma_1^2)}{\N(t; 0, \sigma_1^2)}
    \right)^\alpha
    \right] \label{eq:globalsen} \\
    &= \rdpfunc \left( \epsRgau^{(\sigma_1)}(\alpha) \right)
\end{align}
where (\ref{eq:translation-inv}) is due to the translation invariance of \Renyi divergence, (\ref{eq:rotation}) is due to the rotation trick, (\ref{eq:globalsen}) is because of $\norm{f_1(S') - f_1(S)} \le 1$. 

We now analyze the upper bound of $\E_{t \sim \mu|s} \left[\left(\frac{\mu'(t|s)}{\mu(t|s)}\right)^\alpha  \right]$ when $s > B$ by considering two separate cases: when both $\localSen_{f_2}(S)>\tau$ and $\localSen_{f_2}(S')>\tau$, and when there is at least one of $\localSen_{f_2}(S)$ and $\localSen_{f_2}(S')$ is greater than $\tau$. 

\textbf{Case 1: both $\localSen_{f_2}(S)$ and $\localSen_{f_2}(S')$ are greater than $\tau$. }
In this case, the only known upper bound of $\norm{f_2(S) - f_2(S')}$ is the global sensitivity $\globalSen_{f_2} = 1$. Therefore, we only have $\E_{t \sim \mu|s} \left[\left(\frac{\mu'(t|s)}{\mu(t|s)}\right)^\alpha  \right] \le \rdpfunc \left( \epsRgau^{(\sigma_2)}(\alpha) \right)$ when $s > B$. 
Therefore, in this case we have 
\begin{align}
    &\E_{t \sim \mu|s} \left[
    \left(\frac{\mu'(t|s)}{\mu(t|s)}\right)^\alpha  \right]\Ind[s \le B] +  \E_{t \sim \mu|s} \left[
    \left(\frac{\mu'(t|s)}{\mu(t|s)}\right)^\alpha \right]\Ind[s > B] \\
    &= \rdpfunc \left( \epsRgau^{(\sigma_1)}(\alpha) \right)\Ind[s \le B]+\rdpfunc \left( \epsRgau^{(\sigma_2)}(\alpha) \right)\Ind[s > B]
\end{align}

However, note that when both $\localSen_{f_2}(S)$ and $\localSen_{f_2}(S')$ is greater than $\tau$, we have $\Delta(S) = \Delta(S') = 0$, \rebuttal{which means that there is no privacy loss by releasing the result of $\widehat \Delta(S)$ or $\widehat \Delta(S’)$.} 
Therefore, we have $\mu(s) = \mu'(s) = \lap(s; 0, b)$, and thus 
\begin{align}
    &\E_{(s, t) \sim \mu} \left[\left(\frac{\mu'(s, t)}{\mu(s, t)}\right)^\alpha \right] \\
    &= 
    \E_{s \sim \mu} \left[ \rdpfunc \left( \epsRgau^{(\sigma_1)}(\alpha) \right)\Ind[s \le B]+\rdpfunc \left( \epsRgau^{(\sigma_2)}(\alpha) \right)\Ind[s > B] \right] \\
    &= \rdpfunc \left( \epsRgau^{(\sigma_1)}(\alpha) \right)\Pr[\lap(0, b) \le B]+\rdpfunc \left( \epsRgau^{(\sigma_2)}(\alpha) \right)\Pr[\lap(0, b) > B] \\
    &= (1-\delta_0) \rdpfunc \left( \epsRgau^{(\sigma_1)}(\alpha) \right)+ \delta_0 \rdpfunc \left( \epsRgau^{(\sigma_2)}(\alpha) \right)
\end{align}

\textbf{Case 2: at least one of $\localSen_{f_2}(S)$ and $\localSen_{f_2}(S')$ are smaller than $\tau$.} 
In this case, we know that we have $\norm{f_2(S) - f_2(S')} \le \tau$. Thus, when $s \ge B$, we have 
\begin{align}
    \E_{t \sim \mu|s} \left[
    \left(\frac{\mu'(t|s)}{\mu(t|s)}\right)^\alpha  \right]
    &= \E_{t \sim \N(0, \sigma_2^2)} \left[ 
    \left(
    \frac{\N(t; \norm{f_2(S') - f_2(S)}, \sigma_2^2)}{\N(t; 0, \sigma_2^2)}
    \right)^\alpha
    \right] \\
    &\le \E_{t \sim \N(0, \sigma_2^2 )} \left[ 
    \left(
    \frac{\N(t; \tau, \sigma_2^2)}{\N(t; 0, \sigma_2^2)}
    \right)^\alpha
    \right] \\
    &= \rdpfunc \left( \epsRgau^{(\sigma_2/\tau)}(\alpha) \right)
\end{align}
Thus, we have 
\begin{align}
    &\E_{(s, t) \sim \mu} \left[\left(\frac{\mu'(s, t)}{\mu(s, t)}\right)^\alpha \right] \\
    &\le 
    \E_{s \sim \mu} \left[ \left(\frac{\mu'(s)}{\mu(s)}\right)^\alpha
    \left( \rdpfunc \left( \epsRgau^{(\sigma_1)}(\alpha) \right)\Ind[s \le B]+\rdpfunc \left( \epsRgau^{(\sigma_2/\tau)}(\alpha) \right)\Ind[s > B] \right) \right] \\
    &= \E_{s \sim \mu} \left[ \left(\frac{\mu'(s)}{\mu(s)}\right)^\alpha
    \rdpfunc \left( \epsRgau^{(\sigma_1)}(\alpha) \right) \right] \label{eq:cond} \\
    &= \rdpfunc \left( \epsRgau^{(\sigma_1)}(\alpha) \right) \E_{s \sim \mu} \left[ \left(\frac{\mu'(s)}{\mu(s)}\right)^\alpha \right] \\
    &\le \rdpfunc \left( \epsRgau^{(\sigma_1)}(\alpha) \right) \rdpfunc( \epsRlap^{(b)}(\alpha) )
\end{align}
where (\ref{eq:cond}) is because by our condition, $\sigma_1 = \sigma_2/\tau$. 

Therefore, we have 
\begin{align}
    &D_\alpha( \mu' \| \mu ) \nonumber \\
    &\le \frac{1}{\alpha-1} \log \left(
    \max( (1-\delta_0) \rdpfunc \left( \epsRgau^{(\sigma_1)}(\alpha) \right)+ \delta_0 \rdpfunc \left( \epsRgau^{(\sigma_2)}(\alpha) \right) , \rdpfunc \left( \epsRgau^{(\sigma_1)}(\alpha) \right) \rdpfunc( \epsRlap^{(b)}(\alpha) ) )
    \right) \nonumber \\
    &= \max \left( \frac{1}{\alpha-1} \log \left( (1-\delta_0) \rdpfunc \left( \epsRgau^{(\sigma_1)}(\alpha) \right)+ \delta_0 \rdpfunc \left( \epsRgau^{(\sigma_2)}(\alpha) \right) \right), \epsRgau^{(\sigma_1)}(\alpha) + \epsRlap^{(b)}(\alpha) \right) \nonumber
\end{align}
Since we did not use any condition that depends on the fact that $S' = S \cup \{x\}$, we know that $D_\alpha( \mu \| \mu' )$ also has the exactly the same upper bound, which leads to the conclusion. 
\end{proof}


\newpage

\subsubsection{Discussion: can we improve privacy analysis by not releasing $\widehat \Delta$?}

One may wonder if we can further improve the privacy analysis of PTR by not releasing $\widehat \Delta$. 
However, 
releasing $\widehat \Delta$ is essential for the applications of PTR. The rationale behind PTR is to exploit the fact that, while a function’s global sensitivity may be large, its local sensitivity may be much smaller for most of the “common inputs”. Thus, such a mechanism will only be preferred over a regular output perturbation mechanism when the local sensitivity of data drawn from input data distribution rarely exceeds the threshold. Without knowing about $\widehat \Delta$, the user \textbf{cannot} know whether they are actually enjoying the benefits from PTR or simply wasting privacy budgets on private sensitivity tests. Furthermore, the user cannot adjust the hyperparameters or switch algorithms accordingly. Notably, in Section 5 (the application of PTR in privatizing robust SGD), we also use the information from $\widehat \Delta$ to dynamically adjust the number of gradients to be trimmed (note that this does not affect privacy analysis since the adjustment is post-processing of $\widehat \Delta$). 

Besides, we gave an attempt to directly analyze the variant of PTR that does not release $\hat \Delta$, and we do not see an easy way to obtain a better privacy bound than we have in Theorem 4.3.

We follow the same notations as in the proof of Theorem 4.3:
Given a pair of neighboring dataset $S, S'$, we denote the density of $\M(S)$ as $\mu$ and that of $\M(S')$ as $\mu'$. 
Given $s \sim \widehat \Delta(S)$, we denote $\mu(t|s \le B)$ the density of $\N(f_1(S), \sigma_1^2)$, and $\mu(t|s>B)$ the density of $\N(f_2(S), \sigma_2^2)$. $\mu'(t|s \le B)$ and $\mu'(t|s>B)$ are defined analogously. 

Similar to the proof of Theorem 4.3, we consider two separate cases: when both $\localSen_{f_2}(S)>\tau$ and $\localSen_{f_2}(S')>\tau$, and when there is at least one of $\localSen_{f_2}(S)$ and $\localSen_{f_2}(S')$ is greater than $\tau$. 

\textbf{Case 1. }
For the case that both $\localSen_{f_2}(S)$ and $\localSen_{f_2}(S')$ are greater than $\tau$, from the proof of Theorem 4.3 we know that $\widehat \Delta(S)$ and $\widehat \Delta(S')$ has exactly the same distribution since $\Delta(S)=\Delta(S')=0$. Thus, the exactly the same proof in Theorem 4.3 applies for the case of not releasing $\widehat \Delta$. 

\textbf{Case 2. }
For the case that at least one of $\localSen_{f_2}(S)$ and $\localSen_{f_2}(S')$ is smaller than $\tau$, here's our attempt:
\begin{align}
    E_\alpha(\M(S) \| \M(S')) 
    &= \E_{t \sim \mu} \left[ \left( \frac{\mu'(t)}{\mu(t)} \right)^\alpha \right] \\
    &= \E_{t \sim \mu} \left[ 
    \left( 
    \frac{ \mu'(t|s\le B)\Pr[\widehat \Delta(S') \le B] + \mu'(t|s>B)\Pr[\widehat \Delta(S') > B] }{ \mu(t|s\le B)\Pr[\widehat \Delta(S) \le B] + \mu(t|s>B)\Pr[\widehat \Delta(S) > B] } \right)^\alpha \right] \label{eq:expand}
\end{align}
As we can see, while the distribution of $\M(S)$ is a Gaussian mixture, the probability for different components is also depending on $S$, which introduce more challenge in bounding $E_\alpha(\M(S) \| \M(S'))$. 
One relatively simple way to bound the above expression is by noticing that since $\widehat{\Delta} = \Delta + \mathrm{Lap}(0, b)$, by the privacy guarantee of Laplace mechanism 
we have $\Pr[\widehat \Delta(S') \le B] \le e^{1/b} \Pr[\widehat \Delta(S) \le B]$ and $\Pr[\widehat \Delta(S') > B] \le e^{1/b} \Pr[\widehat \Delta(S) > B]$. Thus, we have
\begin{align}
   (\ref{eq:expand}) 
   &\le \exp \left(\frac{\alpha}{b}\right) \E_{t \sim \mu} \left[ 
    \left( 
    \frac{ \mu'(t|s\le B)\Pr[\widehat \Delta(S) \le B] + \mu'(t|s>B)\Pr[\widehat \Delta(S) > B] }{ \mu(t|s\le B)\Pr[\widehat \Delta(S) \le B] + \mu(t|s>B)\Pr[\widehat \Delta(S) > B] } \right)^\alpha \right] \\
    &\le \exp \left(\frac{\alpha}{b}\right) 
    \E_{t \sim \mu} \left[ 
    \left( 
    \frac{ \mu'(t|s\le B) }{ \mu(t|s\le B) } \right)^\alpha \right]
\end{align}
where the last inequality is due to the quasi-convexity of Renyi divergence \cite{van2014renyi} (note that $\E_{t \sim \mu} \left[ \left( \frac{ \mu'(t|s\le B) }{ \mu(t|s\le B) } \right)^\alpha \right] = \E_{t \sim \mu} \left[ \left( \frac{ \mu'(t|s > B) }{ \mu(t|s > B) } \right)^\alpha \right]$ by construction for this case). 
Thus, we have 
\begin{align}
    R_\alpha(\M(S) \| \M(S')) \le \epsRgau^{(\sigma_1)}(\alpha) + \frac{1}{\alpha-1} \left(\frac{\alpha}{b}\right)
\end{align}
for the case of at least one of $\localSen_{f_2}(S)$ and $\localSen_{f_2}(S')$ are smaller than $\tau$.

Now we show that this bound is not as good as the corresponding bound in Theorem 4.3. 
The corresponding bound in Theorem 4.3 for this case is $\epsRgau^{(\sigma_1)}(\alpha) + \epsRlap^{(b)}(\alpha)$, so we only need to show $\epsRlap^{(b)}(\alpha) < \frac{1}{\alpha-1} \left(\frac{\alpha}{b}\right)$.
\begin{align}
    \epsRlap^{(b)}(\alpha) 
    &= 
    \frac{1}{\alpha-1} \log \left( \frac{\alpha}{2\alpha-1}\exp \left(\frac{\alpha-1}{b}\right) + \frac{\alpha-1}{2\alpha-1} \exp \left(-\frac{\alpha}{b}\right)\right) \\
    &< \frac{1}{\alpha-1} \log \left( \frac{\alpha}{2\alpha-1}\exp \left(\frac{\alpha}{b}\right) + \frac{\alpha-1}{2\alpha-1} \exp \left(\frac{\alpha}{b}\right)\right) \\
    &= \frac{1}{\alpha-1} \log \left( \exp \left(\frac{\alpha}{b}\right) \right) \\
    &= \frac{1}{\alpha-1} \left(\frac{\alpha}{b}\right)
\end{align}
where the first inequality is due to $\exp \left(\frac{\alpha-1}{b}\right) < \exp \left(\frac{\alpha}{b}\right)$ and $\exp \left(-\frac{\alpha}{b}\right) < \exp \left(\frac{\alpha}{b}\right)$. 

Thus, we think at least there are no simple solution for improving the privacy bound for PTR by not releasing $\widehat \Delta$. 
However, even if there are a better way to derive the privacy bound, this variant of PTR may not be user-friendly as the counterpart who release $\widehat \Delta$.

\subsection{\Renyi DP for Subsampled Propose-Test-Release}

We consider two adjacent datasets $S, S'$ where $S' = S \cup \{x\}$. We denote the threshold $B = \log(1/(2\delta_0))b$. Note that we have $\Pr[\lap(0, b) > B] = \delta_0$. The output of Algorithm \ref{alg:ptr-general} on dataset $S$ is a sample from a joint distribution $(\widehat \Delta, \M)(S)$ where $\widehat \Delta(S) = \lap(\Delta(S), b)$ and $\M(S)|_{\widehat \Delta} = \N(f_1(S), \sigma_1^2) \Ind[\widehat \Delta \le B] + \N(f_2(S), \sigma_2^2) \Ind[\widehat \Delta > B]$. 

\begin{customthm}{\ref{thm:subsampled-ptr}}[full version]
\label{thm:subsampled-ptr-appendix}
Let $q$ be the subsampling probability. 
Suppose $\globalSen_{f_1} = \globalSen_{f_2} = 1$ and $\sigma_1 = \sigma_2 / \tau$. 
If $q \le \frac{\exp(-1/b)}{4+\exp(-1/b)}$ and $\sigma_1 \ge \sigma_2 \ge 4$, and $\alpha$ satisfy
$
1<\alpha \leq \frac{1}{2} \sigma^{2}_2 L-2 \ln \sigma_2, 
\alpha \leq \frac{\frac{1}{2} \sigma^{2}_2 L^{2}-\ln 5-2 \ln \sigma_2}{L+\ln (q' \alpha)+1 /\left(2 \sigma^{2}_2\right)},
$
where $L=\ln \left(1+\frac{1}{q'(\alpha-1)}\right)$ and $q' = \frac{q}{q+(1-q)\exp(-1/b)}$, we have 
\begin{align}
    \eps_{\ptr \circ \poisson}(\alpha) 
    \le \rdpfunc^{-1} \left( \max( \worst, \easy, \hard ) \right) \nonumber
\end{align}
where 
\begin{align}
    \worst &= 1 + 2q^2 \alpha (\alpha-1) \left( \frac{1-\delta_0}{\sigma_1^2} + \frac{\delta_0}{\sigma_2^2} \right) \\
    \easy &= \rlapin{\alpha} + \frac{2\alpha (\alpha-1)}{\sigma_1^2} \left[ \rlapin{\alpha} - 2(1-q)\rlapin{\alpha-1} + (1-q)^2 \rlapin{\alpha-2} \right] \\
    \hard &= \rlaptin{\alpha} + \frac{2\alpha (\alpha-1)}{\sigma_1^2} \left[ \rlaptin{\alpha} - 2(1-q)\rlaptin{\alpha+1} + (1-q)^2 \rlaptin{\alpha+2} \right]
\end{align}
with $\rlapin{\alpha} = \E_{s \sim \mu_0} \left[\left(\frac{\mu(s)}{\mu_0(s)}\right)^\alpha \right]$
and 
$\rlaptin{\alpha} = \E_{s \sim \mu} \left[\left(\frac{\mu_0(s)}{\mu(s)}\right)^\alpha \right]$
for $\mu_0 \sim \lap(0, b)$ and $\mu \sim (1-q) \lap(0, b) + q\lap(1, b)$. 
\end{customthm}

\begin{proof}
Let $T$ denote a set-valued random variable defined by taking a random subset of $S$, where each element of $S$ is independently placed in $T$ with probability $q$. Conditioned on $T$, the PTR outputs $(\widehat \Delta, \M)(T)$. Thus, 
\begin{align}
    \left( \widehat \Delta, \M \right)(S) 
    &= \sum_{T \subseteq S} p_T \cdot \left(\widehat \Delta, \M\right)(T) \\
    \left( \widehat \Delta, \M \right)(S') &= \sum_{T \subseteq S} p_T \cdot \left( (1-q) \cdot (\widehat \Delta, \M)(T) + q \cdot (\widehat \Delta, \M)(T\cup \{x\}) \right) 
\end{align}
where $p_T$ denotes the probabilty of sampling the subset $T$. 

\begin{align}
&D_\alpha \left(\left( \widehat \Delta, \M \right)(S') \Vert \left( \widehat \Delta, \M \right)(S)\right) \nonumber \\
&= D_\alpha \left(
\sum_{T} p_T \cdot \left( (1-q) \cdot (\widehat \Delta, \M)(T) + q \cdot (\widehat \Delta, \M)(T\cup \{x\}) \right)
\Vert 
\sum_T p_T \cdot \left(\widehat \Delta, \M\right)(T)
\right) \nonumber \\
&\le \sup_T 
D_\alpha \left(
(1-q) \cdot (\widehat \Delta, \M)(T) + q \cdot (\widehat \Delta, \M)(T\cup \{x\})
\Vert 
\left(\widehat \Delta, \M\right)(T)
\right) \nonumber
\end{align}
where the last step is due to the quasi-convexity of \Renyi divergence (\cite{van2014renyi}, Theorem 13). Symmetrically, we also have 
\begin{align}
    &D_\alpha \left(\left( \widehat \Delta, \M \right)(S') \Vert \left( \widehat \Delta, \M \right)(S)\right) \\
    &\le 
    \sup_T 
D_\alpha \left( 
\left(\widehat \Delta, \M\right)(T) \Vert
(1-q) \cdot (\widehat \Delta, \M)(T) + q \cdot (\widehat \Delta, \M)(T\cup \{x\})
\right)
\end{align}

Fix a subset $T$ and denote $T' = T \cup \{x\}$. 
We use $\mu_0$ to denote the density function of $(\widehat \Delta, \M)(T)$, where $\mu_0(s, t)$ refers to the density on $(s, t)$. 
We use $\mu_1$ to denote the density function of $(\widehat \Delta, \M)(T')$, where $\mu_1(s, t)$ refers to the density on $(s, t)$. 
Let $\mu = (1-q)\mu_0 + q\mu_1$. 
We want to bound 
$\E_{\mu_0} \left[ \left(\frac{\mu}{\mu_0} \right)^\alpha \right]$ and $\E_{\mu} \left[ \left(\frac{\mu_0}{\mu}\right)^\alpha \right]$. 

We first bound $\E_{\mu_0} \left[ \left(\frac{\mu}{\mu_0} \right)^\alpha \right]$, which is usually considered as an easier one. 

By decomposition, we have 
\begin{align}
    \E_{s, t \sim \mu_0} \left[ \left(\frac{\mu(s, t)}{\mu_0(s, t)}\right)^\alpha \right] 
    &= \E_{s, t \sim \mu_0} \left[ \left(\frac{\mu(s)\mu(t|s)}{\mu_0(s)\mu_0(t|s)}\right)^\alpha \right]  \\
    &= \E_{s \sim \mu_0} \left[
    \left(\frac{\mu(s)}{\mu_0(s)}\right)^\alpha
    \E_{t \sim \mu_0(\cdot|s)} \left[\left(\frac{\mu(t|s)}{\mu_0(t|s)}\right)^\alpha\right]
    \right] \label{eq:all}
\end{align}

For the density of conditional distribution $\mu(t|s)$, we have
\begin{align}
    \mu(t|s) &= \frac{\mu(s, t)}{\mu(s)} \\
    &= \frac{(1-q)\mu_0(s, t) + q\mu_1(s, t)}{\mu(s)} \\
    &= 
    I[s \le B] \cdot \frac{ (1-q) \mu_0(s) \cdot \N(t ; f_1(T), \sigma_1^2) + q\mu_1(s) \cdot \N(t; f_1(T'), \sigma_1^2) }{\mu(s)} \\
    &~~~~~~+  I[s > B] \cdot \frac{ (1-q) \mu_0(s) \cdot \N(t; f_2(T), \sigma_2^2) + q\mu_1(s) \cdot \N(t; f_2(T'), \sigma_2^2) }{\mu(s)}
\end{align}

Denote $A(s) = \frac{q\mu_1(s)}{\mu(s)}$. Recall that $\mu(s) = (1-q)\mu_0(s)+q\mu_1(s)$, so we have $\frac{(1-q)\mu_0(s)}{\mu(s)} = 1 - A$. 
Then we have  
\begin{align}
    \mu(t|s) &= I[s \le B] \left( (1-A) \cdot \N(t ; f_1(T), \sigma_1^2) + A \cdot \N(t ; f_1(T'), \sigma_1^2) \right) \\
    &~~~+ I[s > B] \left( (1-A) \cdot \N(t ; f_2(T), \sigma_2^2) + A \cdot \N(t ; f_2(T'), \sigma_2^2) \right)
\end{align}
and we know that 
\begin{align}
    \mu_0(t|s) = I[s \le B] \N(t ; f_1(T), \sigma_1^2) + 
    I[s > B] \N(t ; f_2(T), \sigma_2^2)
\end{align}
Therefore we have 
\begin{align}
    &\E_{t \sim \mu_0(\cdot|s)} \left[\left(\frac{\mu(t|s)}{\mu_0(t|s)}\right)^\alpha\right] \\
    &= \E_{t \sim \mu_0(\cdot|s)} \left[\left(\frac{\mu(t|s)}{\mu_0(t|s)}\right)^\alpha I[s \le B] + \left(\frac{\mu(t|s)}{\mu_0(t|s)}\right)^\alpha I[s > B]
    \right]\\
    &= \E_{t \sim \mu_0(\cdot|s)} \left[\left(\frac{\mu(t|s)}{\mu_0(t|s)}\right)^\alpha\right]I[s \le B] + \E_{t \sim \mu_0(\cdot|s)} \left[\left(\frac{\mu(t|s)}{\mu_0(t|s)}\right)^\alpha\right]I[s > B] \\
    &= I[s \le B] \E_{t \sim \N(f_1(T), \sigma_1^2)} \left[ \left( (1-A) + A \cdot \frac{\N(t ; f_1(T'), \sigma_1^2)}{\N(t ; f_1(T), \sigma_1^2)} \right)^\alpha \right] \\
    &~~~+ 
    I[s > B] \E_{t \sim \N(f_2(T), \sigma_2^2)} \left[ \left( (1-A) + A \cdot \frac{\N(t ; f_2(T'), \sigma_2^2)}{\N(t ; f_2(T), \sigma_2^2)} \right)^\alpha \right] 
\end{align}
Note that $\E_{t \sim \N(f_1(T), \sigma_1^2)} \left[ \left( (1-A) + A \cdot \frac{\N(t ; f_1(T'), \sigma_1^2)}{\N(t ; f_1(T), \sigma_1^2)} \right)^\alpha \right]$ can be exactly bounded by the \Renyi DP of subsampled RDP with sampling probability $A$. 
\begin{lemma}[\cite{mironov2019r}, Theorem 11]
If $q \leq \frac{1}{5}, \sigma \geq 4$, and $\alpha$ satisfy $1<\alpha \leq \frac{1}{2} \sigma^{2} L-2 \ln \sigma$, $\alpha \leq \frac{\frac{1}{2} \sigma^{2} L^{2}-\ln 5-2 \ln \sigma}{L+\ln (q \alpha)+1 /\left(2 \sigma^{2}\right)}$ where $L=\ln \left(1+\frac{1}{q(\alpha-1)}\right)$, then for any function $f$ with $\ell_{2}$-sensitivity $\tau$ satisfies 
\begin{align}
    \E_{t \sim \N(f(T), \sigma^2)} \left[ \left( (1-q) + q \cdot \frac{\N(t ; f(T'), \sigma^2)}{\N(t ; f_2(T), \sigma^2)} \right)^\alpha \right]
    \le 1 + 2q^2 \tau^2 \alpha (\alpha-1) / \sigma^2
\end{align}
\label{lemma:gaussian-subsample}
\end{lemma}
Similar to the proof of Theorem \ref{thm:rdp-for-ptr-maintext}, we consider two cases:

\textbf{Case 1: both $\localSen_{f_2}(T)$ and $\localSen_{f_2}(T')$ are greater than $\tau$. }
In this case, the only known upper bound of $\norm{f_2(T) - f_2(T')}$ is the global sensitivity $\globalSen_{f_2} = 1$. 
Therefore, we only have $\E_{t \sim \mu|s} \left[\left(\frac{\mu'(t|s)}{\mu(t|s)}\right)^\alpha  \right] \le 1 + 2A^2 \alpha (\alpha-1) / \sigma_2^2$ when $s > B$. 
Therefore, in this case we have 
\begin{align}
    &\E_{t \sim \mu|s} \left[
    \left(\frac{\mu'(t|s)}{\mu(t|s)}\right)^\alpha  \right]\Ind[s \le B] +  \E_{t \sim \mu|s} \left[
    \left(\frac{\mu'(t|s)}{\mu(t|s)}\right)^\alpha \right]\Ind[s > B] \\
    &= (1 + 2A^2 \alpha (\alpha-1) / (\sigma_1)^2)
    \Ind[s \le B]+(1 + 2A^2 \alpha (\alpha-1) / (\sigma_2)^2)\Ind[s > B] \\
    &= 1 + 2A^2 \alpha (\alpha-1) \left( \frac{\Ind[s \le B]}{\sigma_1^2} + \frac{\Ind[s > B]}{\sigma_2^2} \right)
\end{align}

However, note that when both $\localSen_{f_2}(S)$ and $\localSen_{f_2}(S')$ is greater than $\tau$, we have $\Delta(S) = \Delta(S') = 0$. Therefore, we have $\mu(s) = \mu'(s) = \lap(s; 0, b)$, $A = q$, and thus
\begin{align}
    &\E_{(s, t) \sim \mu} \left[\left(\frac{\mu'(s, t)}{\mu(s, t)}\right)^\alpha \right] \\
    &= 
    \E_{s \sim \mu} \left[ 1 + 2q^2 \alpha (\alpha-1) \left( \frac{\Ind[s \le B]}{\sigma_1^2} + \frac{\Ind[s > B]}{\sigma_2^2} \right) \right] \\
    &= 1 + 2q^2 \alpha (\alpha-1) \left( \frac{1-\delta_0}{\sigma_1^2} + \frac{\delta_0}{\sigma_2^2} \right)
\end{align}

\textbf{Case 2: at least one of $\localSen_{f_2}(T)$ and $\localSen_{f_2}(T')$ are smaller than $\tau$.} 
In this case, we know that we have $\norm{f_2(T) - f_2(T')} \le \tau$. 
Since $A = \frac{q\mu_1(s)}{(1-q)\mu_0(s)+q\mu_1(s)} \le \frac{q}{(1-q)+q\exp(-1/b)}$ which satisfy the conditions in Lemma \ref{lemma:gaussian-subsample} by our assumption, 
when $s \ge B$, we have $\E_{t \sim \mu|s} \left[\left(\frac{\mu'(t|s)}{\mu(t|s)}\right)^\alpha  \right] \le 1 + 2A^2 \alpha (\alpha-1) / (\sigma_2/\tau)^2$. 
Thus we have 
\begin{align}
    &\E_{s, t \sim \mu_0} \left[ \left(\frac{\mu(s, t)}{\mu_0(s, t)}\right)^\alpha \right] \\
    &= \E_{s \sim \mu_0} \left[
    \left(\frac{\mu(s)}{\mu_0(s)}\right)^\alpha
    \E_{t \sim \mu_0(\cdot|s)} \left[\left(\frac{\mu(t|s)}{\mu_0(t|s)}\right)^\alpha\right]
    \right] \\
    &\le \E_{s \sim \mu_0} \left[
    \left(\frac{\mu(s)}{\mu_0(s)}\right)^\alpha
    \left(
    1 + 2A^2 \alpha (\alpha-1) \left( \frac{\Ind[s \le B]}{\sigma_1^2} + \frac{\Ind[s > B]}{(\sigma_2/\tau)^2} \right)
    \right)
    \right] \\
    &= \E_{s \sim \mu_0} \left[
    \left(\frac{\mu(s)}{\mu_0(s)}\right)^\alpha \right] 
    + 2\alpha(\alpha-1) 
    \E_{s \sim \mu_0} \left[
    \left(\frac{\mu(s)}{\mu_0(s)}\right)^\alpha A^2 \left( \frac{I[s \le B]}{\sigma_1^2} + \frac{I[s > B]}{(\sigma_2/\tau)^2} \right) \right] \label{eq:Asquare} \\
    &= \E_{s \sim \mu_0} \left[
    \left(\frac{\mu(s)}{\mu_0(s)}\right)^\alpha \right] 
    + \frac{ 2\alpha(\alpha-1) }{ \sigma_1^2 } 
    \E_{s \sim \mu_0} \left[
    \left(\frac{\mu(s)}{\mu_0(s)}\right)^\alpha A^2  \right]
\end{align}

Denote $\rlapin{\alpha} = \E_{s \sim \mu_0} \left[\left(\frac{\mu(s)}{\mu_0(s)}\right)^\alpha \right]$, which is the RDP of subsampled Laplace mechanism with sampling rate $q$ (note that $\mu(s) = (1-q)\mu_0(s)+q\mu_1(s)$). 

Since 
\begin{align}
    A &= \frac{q\mu_1(s)}{\mu(s)} = 1 - \frac{(1-q)\mu_0(s)}{\mu(s)} \\
    A^2 &= 1 - \frac{2(1-q)\mu_0(s)}{\mu(s)} + \frac{(1-q)^2\mu_0^2(s)}{\mu^2(s)}
\end{align}
Plug this back to the second term, we have 
\begin{align}
    &\E_{s \sim \mu_0} \left[
    \left(\frac{\mu(s)}{\mu_0(s)}\right)^\alpha A^2  \right] \\
    &= \E_{s \sim \mu_0} \left[
    \left(\frac{\mu(s)}{\mu_0(s)}\right)^\alpha \left(1 - \frac{2(1-q)\mu_0(s)}{\mu(s)} + \frac{(1-q)^2\mu_0^2(s)}{\mu^2(s)}\right) \right] \\
    &= \E_{s \sim \mu_0} \left[
    \left(\frac{\mu(s)}{\mu_0(s)}\right)^\alpha \right] -2(1-q) \E_{s \sim \mu_0} \left[
    \left(\frac{\mu(s)}{\mu_0(s)}\right)^{\alpha-1} \right] +(1-q)^2 \E_{s \sim \mu_0} \left[
    \left(\frac{\mu(s)}{\mu_0(s)}\right)^{\alpha-2} \right]  \\
    &= \rlapin{\alpha} - 2(1-q)\rlapin{\alpha-1} + (1-q)^2 \rlapin{\alpha-2}
\end{align}
Note that this bound is independent on $T$ due to translation invariance, hence it is an upper bound for arbitrary $T$ (that satisfy case 2). Thus, the overall bound becomes
\begin{align}
    \E_{s, t \sim \mu_0} \left[ \left(\frac{\mu(s, t)}{\mu_0(s, t)}\right)^\alpha \right] 
    \le \rlapin{\alpha} + 
    \frac{2\alpha (\alpha-1)}{\sigma_1^2} \left[ \rlapin{\alpha} - 2(1-q)\rlapin{\alpha-1} + (1-q)^2 \rlapin{\alpha-2} \right] \label{eq:bound-easy}
\end{align}

Denote $\rlaptin{\alpha} = \E_{s \sim \mu} \left[ \left(\frac{\mu_0(s)}{\mu(s)}\right)^\alpha \right]$. 
Since we know that 
\begin{align}
    \rlaptin{\alpha}
    = \E_{\mu} \left[ \left(\frac{\mu_0}{\mu}\right)^\alpha \right] 
    = \E_{\mu_0} \left[ \left(\frac{\mu_0}{\mu}\right)^{\alpha-1} \right]
    = \E_{\mu_0} \left[ \left(\frac{\mu}{\mu_0}\right)^{1-\alpha} \right]
    = \rlapin{1-\alpha}
\end{align}
Thus, by setting $\alpha \leftarrow 1-\alpha$, we have 
\begin{align}
    &\E_{s, t \sim \mu} \left[ \left(\frac{\mu_0(s, t)}{\mu(s, t)}\right)^\alpha \right] \\
    &= \E_{s, t \sim \mu_0} \left[ \left(\frac{\mu(s, t)}{\mu_0(s, t)}\right)^{1-\alpha} \right] \\
    &= \rlaptin{1-\alpha} + 
    \frac{2\alpha (\alpha-1)}{\sigma_1^2} \left[ \rlaptin{1-\alpha} - 2(1-q)\rlaptin{-\alpha} + (1-q)^2 \rlaptin{-\alpha-1} \right] \\
    &= \rlaptin{\alpha} + 
    \frac{2\alpha (\alpha-1)}{\sigma_1^2} \left[ \rlaptin{\alpha} - 2(1-q)\rlaptin{\alpha+1} + (1-q)^2 \rlaptin{\alpha+2} \right]
\end{align}

So overall, the RDP of subsampled PTR is 
\begin{align}
    \frac{1}{\alpha-1} \ln &\left(
    \max \left(1 + 2q^2 \alpha (\alpha-1) 
    \left( \frac{1-\delta_0}{\sigma_1^2} + \frac{\delta_0}{\sigma_2^2} \right), \right.\right. \\
    &\rlapin{\alpha} + \frac{2\alpha (\alpha-1)}{\sigma_1^2} \left[ \rlapin{\alpha} - 2(1-q)\rlapin{\alpha-1} + (1-q)^2 \rlapin{\alpha-2} \right], \\
    &\rlaptin{\alpha} + \left.\left. \frac{2\alpha (\alpha-1)}{\sigma_1^2} \left[ \rlaptin{\alpha} - 2(1-q)\rlaptin{\alpha+1} + (1-q)^2 \rlaptin{\alpha+2} \right]
    \right)
    \right)
\end{align}
\end{proof}

\subsection{Bound of the Local sensitivity after $r$ adding/removal for Trimmed Sum}
Recall that we denote a dataset $S = \{x_1, \ldots, x_m\}$, and $x_{(k)}$ denote the $k$th smallest data point among $S$ in $\ell_2$ norm, i.e., $x_{(1)} \le x_{(2)} \le \ldots \le x_{(m)}$. 
$\tmean_\nfail(S) = \sum_{i=1}^{m-\nfail} x_{(i)}$ if $m > \nfail$, or 0 if $m \le \nfail$. 

\begin{customthm}{\ref{thm:localsenr}}[Restate]
$\localSen_{\tmean_{\nfail}}^{(r)}(S) = \norm{x_{(m-F+1+r)}}$ if $r \le \nfail-1$, or $\globalSen_{\tmean_{\nfail}}$ if $r > \nfail-1$. 
\end{customthm}
\begin{proof}
The $\globalSen_{\tmean_{\nfail}}$ for $r > \nfail-1$ is trivial as the local sensitivity can never be larger than global sensitivity. 
When $r \le \nfail-1$, it is easy to see that the local sensitivity of $\tmean_\nfail$ is just $x_{(m-F+1)}$, as we can add the element with the maximum possible norm $x_\infty$ in the data space to $S$, so that $\tmean_\nfail(S \cup \{x_\infty\}) = \sum_{i=1}^{m+1-\nfail} x_{(i)}$, and $\norm{\tmean_\nfail(S \cup \{x_\infty\}) - \tmean_\nfail(S)} = \norm{x_{(i)}}$. 
If the added element has norm smaller than $\norm{x_{(m+1-\nfail)}}$, we will always have $\norm{\tmean_\nfail(S \cup \{x_\infty\}) - \tmean_\nfail(S)} < \norm{x_{(i)}}$. We can easily see that single element removal will also not change $\tmean_\nfail(S)$ that much. Thus $\localSen_{\tmean_{\nfail}}^{(0)}(S) = x_{(m-F+1)}$. 

To maximize the local sensitivity of $S$ with $r$ elements addition/removal, it's trivial to see that the best strategy is simply adding element with the maximum possible norm $x_\infty$ in the data space to $S$, and the local sensitivity for the changed dataset $\Tilde{S} = S \cup (\{x_\infty\} \times r)$ has local sensitivity $\norm{x_{(m+r)-\nfail+1}} = \norm{x_{m-\nfail+1+r}}$ as long as $r \le \nfail-1$.  
\end{proof}

\subsection{Convergence Guarantee of PTR-based Gradient Aggregation under Byzantine Failure}

\newcommand{\mH}{\mathcal{H}^{(t)}}
\newcommand{\tilMH}{\Tilde{\mathcal{H}}^{(t)}}
\newcommand{\B}{\mathcal{B}}
\newcommand{\mL}{\mathcal{L}}

\paragraph{Settings of Robust Training.}
We denote the target loss function as $\mL(w) = \E_{z \sim \D} \left[ \ell(w, z) \right]$, where $w \in \R^d$ is the model parameters and $z$ is a data point randomly drawn from some distribution $\D$. We assume $\ell$ is $R$-Lipschitz, $\beta$-smooth and $\alpha$-strongly convex. We have $n$ stochastic gradient oracles $g_1, \ldots, g_n$, where at each iteration $t$, for every non-corrupted gradient oracles $i$, it is an unbiased estimator $g^{(t)}_i$ for the gradient of the global expected loss function with respect to the current model parameters $w_t$, i.e., $\E[ g^{(t)}_i ] = \g \mL(w^{(t)})$. 
We additionally assume that non-corrupted stochastic gradients have bounded variance, i.e., for some $\sigma > 0$ we have
\begin{align}
    \E_{g_i^{(t)}} \left[ \norm{g^{(t)}_i -  \E[g^{(t)}_i]}^2 \right] \le \sigma^2
\end{align}
at every step $t$. 

\begin{remark}
We do not consider the effect of subsampling here for clean presentation. The effect from subsampling could be easily handled by deriving a high probability upper bound for the number of corrupted gradient oracles that will be sampled. 
\end{remark}

SGD with PTR works as 
\begin{align}
    \widehat{\Delta}, \tilde{g}^{(t)} &\leftarrow \ptr(\{g_1, \ldots, g_n\}) \\
    w^{(t+1)} &\leftarrow w^{(t)} - \left(
    \eta_A I[\widehat{\Delta} < \log(1/(2\delta_0))] + 
    \eta_B I[\widehat{\Delta} \ge \log(1/(2\delta_0))]
    \right)
    \tilde{g}^{(t)}
\end{align}

Further, we call it Routine A if $\widehat \Delta$ is small and $\ptr(\{g_1, \ldots, g_n\}) = \sum_{i=1}^n g^{(t)}_{(i)} + \N (0, \sigma_1\iden_d)$, and call it Routine B if $\widehat \Delta$ is large and $\ptr(\{g_1, \ldots, g_n\}) = \sum_{i=1}^{n-\nfail} g^{(t)}_{(i)} + \N (0, \sigma_2\iden_d)$. 
It makes sense to use a smaller learning rate $\eta_A$ when the PTR goes to Routine A, and use a larger learning rate $\eta_B$ for Routine B, since the two routines use different amount of gradient information. 

\begin{customthm}{\ref{thm:convergence}}[formal version]
Let $w^* \in \argmin_w \mL(w)$. 
If there are at most $\nfail$ gradients being corrupted at each iteration, 
and if we set $\eta_A = \frac{n-F}{n} \eta_B$ and 
$
\sigma_2^2 \ge \frac{(n-\nfail)(n+1)}{n^2} \left( \frac{(n-\nfail)\sigma^2+FR}{d} + \sigma_1^2 \right)
$, then as $t \rightarrow \infty$, we have 
\begin{align}
    \E_{g, \xi} \left[ \norm{w^{(t)} - w^*}^2 \right] 
    &\le \frac{M_B}{1-\rho_B}
\end{align}
for 
\begin{align}
    0 < \eta_B \le \frac{ 2 \alpha (n-2\nfail) }{ n^2 + (n-\nfail+1)(n-\nfail)\beta^2 }
\end{align}
where 
\begin{align}
    \rho_B &= 1 - 2\eta_B \alpha (n-2\nfail) + \eta_B^2 (n^2 + (n-\nfail+1)(n-\nfail)\beta^2 ) \\
    M_B &= \eta_B^2 (n-\nfail+1) \left( (n-\nfail)\sigma^2 + d \sigma_2^2 \right) + \left( \frac{\nfail}{n} \right)^2 R^2
\end{align}

\end{customthm}

\begin{proof}
Let $\eta = \eta_A I[\widehat{\Delta} < \log(1/(2\delta_0))] + \eta_B I[\widehat{\Delta} \ge \log(1/(2\delta_0))]$. 
By the update rule of parameters in SGD, we have
\begin{align}
    \norm{w^{(t+1)} - w^*}^2 = 
    \norm{w^{(t)} - w^*}^2 - 2 \eta 
    \inner{w^{(t)} - w^*}{ \Tilde{g}^{(t)}} + \eta^2 \norm{\Tilde{g}^{(t)}}^2
\end{align}
where 
\begin{align}
    \Tilde{g}^{(t)} = \ptr(\{g_1, \ldots, g_n\})
\end{align}
If $\ptr$ runs Routine B, then 
\begin{align}
    \Tilde{g}^{(t)} = 
    \sum_{i=1}^{n-\nfail} g^{(t)}_{(i)} + \xi, \xi \sim \N(0, \sigma_2 \textbf{I}_d)
\end{align}
If $\ptr$ runs Routine A, then 
\begin{align}
    \Tilde{g}^{(t)} = 
    \sum_{i=1}^{n} g^{(t)}_{(i)} + \xi, \xi \sim \N(0, \sigma_1 \textbf{I}_d)
\end{align}

In the following proof, we consider the two routines separately.  
We denote a set of $n-\nfail$ non-corrupted gradients at step $t$ as $\mH$. 
We use $g^{(t)}_{(i)}$ the $i$th smallest gradient among the set of gradients at step $t$, i.e., $g^{(t)}_{(1)} \le \ldots \le g^{(t)}_{(n)}$. 

\textbf{Case of Running Routine B. }
In this case, $\E [\norm{\xi}^2] = d \sigma_2^2$. 
We first upper bound $\norm{\Tilde{g}^{(t)}}$:
\begin{align}
    \norm{\Tilde{g}^{(t)}} 
    &= \norm{ \sum_{i=1}^{n-\nfail} g^{(t)}_{(i)} + \xi } \\
    &\le \sum_{i=1}^{n-\nfail} \norm{g^{(t)}_{(i)}} + \norm{\xi} \\
    &\le \sum_{i \in \mH} \norm{g^{(t)}_{i}} + \norm{\xi} 
\end{align}
Therefore 
\begin{align}
    \norm{\Tilde{g}^{(t)}}^2 
    \le (n-\nfail+1) \left( \sum_{i \in \mH} \norm{g^{(t)}_{i}}^2 + \norm{\xi}^2 \right)
\end{align}
due to AM-QM inequality. 

Now we lower bound $\inner{w^t - w^*}{ \Tilde{g}^{(t)}}$. We denote the non-corrupted gradients that are not being trimmed in step $t$ as $\tilMH = \mH \cap \{g^{(t)}_{(1)}, \ldots, g^{(t)}_{(n-\nfail)}\}$, and we denote the corrupted gradients that are not being trimmed in step $t$ as $\B^{(t)} = \{g^{(t)}_{(1)}, \ldots, g^{(t)}_{(n-\nfail)}\} \setminus \tilMH$. 
Since there are at most $\nfail$ corrupted gradients, we have $|\tilMH| \ge n-2\nfail$ and $|\B| \le \nfail$. 
Since 
\begin{align}
    \Tilde{g}^{(t)} = 
    \sum_{i \in \tilMH} g^{(t)}_i + \sum_{i \in \B^t} g^{(t)}_i + \xi
\end{align}
we have 
\begin{align}
    \inner{w^{(t)} - w^*}{ \Tilde{g}^{(t)}} &= 
    \sum_{i \in \tilMH} \inner{w^{(t)} - w^*}{ g^{(t)}_i } 
    + \sum_{i \in \B^t} \inner{w^{(t)} - w^*}{ g^{(t)}_i } 
    +  \inner{w^{(t)} - w^*}{ \xi } \\
    &\ge 
    \sum_{i \in \tilMH} \inner{w^{(t)} - w^*}{ g^{(t)}_i } 
    - \nfail R \norm{w^{(t)} - w^*}
    + \inner{w^{(t)} - w^*}{ \xi }
\end{align}

Denote this lower bound as 
\begin{align}
    \phi_t = 
    \sum_{i \in \tilMH} \inner{w^{(t)} - w^*}{ g^{(t)}_i }
    - \nfail R \norm{w^{(t)} - w^*}
    + \inner{w^{(t)} - w^*}{ \xi }
\end{align}
Then 
\begin{align}
    \norm{w^{(t+1)} - w^*}^2 \le 
    \norm{w^{(t)} - w^*}^2 - 2 \eta \phi_t + \eta^2 (n-\nfail+1) \left( \sum_{i \in \mH} \norm{g^{(t)}_{i}}^2 + \norm{\xi}^2 \right)
\end{align}
Now this upper bound only include quantities that are independent from corrupted gradients. 
Take expectation of both sides over $g^{(t)}$ and $\xi$, we have 
\begin{align}
    &\E_{g, \xi} \left[ \norm{w^{(t+1)} - w^*}^2 \right] \\
    &\le \norm{w^{(t)} - w^*}^2 - 2\eta \E_{g, \xi} [\phi_t] + \eta^2 (n-\nfail+1) \E_{g, \xi} \left[ \sum_{i \in \mH} \norm{g^{(t)}_{i}}^2 + \norm{\xi}^2 \right] 
\end{align}

For $\E_{g, \xi} [\phi_t]$, we can obtain its lower bound
\begin{align}
    \E_{g, \xi} [\phi_t] 
    &= \sum_{i \in \tilMH} \inner{w^{(t)} - w^*}{ \g \mL(w^{(t)}) } 
    - \nfail R \norm{w^{(t)} - w^*} \\
    &\ge \alpha (n-2\nfail) \norm{w^{(t)} - w^*}^2 - \nfail R \norm{w^{(t)} - w^*}
\end{align}
since $\mL$ is $\alpha$-strongly convex. 

For $\E_{g, \xi} \left[ \sum_{i \in \mH} \norm{g^{(t)}_{i}}^2 + \norm{\xi}^2 \right]$, 
since $\E_g \left[ \norm{g^{(t)}_{i}}^2 \right] \le 
\sigma^2 + \norm{\g \mL(w^t)}^2$ and $\E [\norm{\xi}^2] = d \sigma_2^2$, 
we have 
\begin{align}
    &\E_{g, \xi} \left[ \sum_{i \in \mH} \norm{g^{(t)}_{i}}^2 + \norm{\xi}^2 \right] \\
    &\le (n-\nfail) \left( \sigma^2 + \norm{\g \mL(w^{(t)})}^2 \right) + d \sigma_2^2 \\
    &\le (n-\nfail) \left( \sigma^2 + \beta^2 \norm{w^{(t)} - w^*}^2 \right) + d \sigma_2^2 
\end{align}
since $\norm{\g \mL(w^{(t)})} = \norm{\g \mL(w^{(t)}) - \g \mL(w^{*})} \le \beta \norm{w^{(t)} - w^*}$ by $\beta$-smoothness of $\mL$.

Plugging in the lower and upper bounds, we have 
\begin{align}
    &\E_{g, \xi} \left[ \norm{w^{(t+1)} - w^*}^2 \right] \\
    &\le 
    \left( 1 - 2\eta \alpha (n-2\nfail) + \eta^2 \beta^2 (n-\nfail+1)(n-\nfail) \right) \norm{w^{(t)} - w^*}^2 \\
    &~~~~+ 2\eta \nfail R \norm{w^{(t)} - w^*} \\
    &~~~~+ \eta^2 (n-\nfail+1) \left( (n-\nfail)\sigma^2 + d \sigma_2^2 \right) \\
    &\le \left( 1 - 2\eta \alpha (n-2\nfail) + \eta^2 (n^2 + (n-\nfail+1)(n-\nfail)\beta^2 ) \right) \norm{w^{(t)} - w^*}^2 \\
    &+ \eta^2 (n-\nfail+1) \left( (n-\nfail)\sigma^2 + d \sigma_2^2 \right) + \left( \frac{\nfail}{n} \right)^2 R^2
\end{align}
where the last step is due to 
\begin{align}
    2\eta \nfail R \norm{w^{(t)} - w^*}
    \le \left( \frac{\nfail}{n} \right)^2 R^2 + n^2 \eta^2 \norm{w^{(t)} - w^*}^2
\end{align}

Let 
\begin{align}
    \rho_B &= 1 - 2\eta \alpha (n-2\nfail) + \eta^2 (n^2 + (n-\nfail+1)(n-\nfail)\beta^2 ) \\
    M_B &= \eta^2 (n-\nfail+1) \left( (n-\nfail)\sigma^2 + d \sigma_2^2 \right) + \left( \frac{\nfail}{n} \right)^2 R^2
\end{align}

Then we have 
\begin{align}
    \E_{g, \xi} \left[ \norm{w^{(t+1)} - w^*}^2 \right] 
    &\le \rho_B \norm{w^{(t)} - w^*}^2 + M_B
\end{align}

Therefore, as long as $\rho_B < 1$, $w^{(t)}$ will eventually 

For $\rho_B < 1$, we need 
\begin{align}
    0 < \eta \le \frac{ 2 \alpha (n-2\nfail) }{ n^2 + (n-\nfail+1)(n-\nfail)\beta^2 }
\end{align}

\textbf{Case of Running Routine A.} We follow a similar analysis as for the case of Routine A. 
In this case, $\E [\norm{\xi}^2] = d \sigma_1^2$. 

We first upper bound $\norm{\Tilde{g}^{(t)}}$:
\begin{align}
    \norm{\Tilde{g}^{(t)}} 
    &= \norm{ \sum_{i=1}^{n} g^{(t)}_{(i)} + \xi } \\
    &\le \sum_{i=1}^{n} \norm{g^{(t)}_{(i)}} + \norm{\xi}
\end{align}
Therefore 
\begin{align}
    \norm{\Tilde{g}^{(t)}}^2 
    \le (n+1) \left( \sum_{i \in \mH} \norm{g^{(t)}_{i}}^2 + \norm{\xi}^2 \right)
\end{align}

In this case, there are no gradients being corrupted, and thus we have $|\tilMH| \ge n-\nfail$. Therefore, for $\E_{g, \xi} [\phi_t]$ we have 
\begin{align}
    \E_{g, \xi} [\phi_t] 
    &= \sum_{i \in \tilMH} \inner{w^{(t)} - w^*}{ \g \mL(w^{(t)}) } 
    - \nfail R\norm{w^{(t)} - w^*} \\
    &\ge \alpha (n-\nfail) \norm{w^{(t)} - w^*}^2 - \nfail R \norm{w^{(t)} - w^*}
\end{align}

For $\E_{g, \xi} \left[ \sum_{i \in \mH} \norm{g^{(t)}_{(i)}}^2 
+ \sum_{i \in [n]\setminus \mH} \norm{g^{(t)}_{(i)}}^2
+ \norm{\xi}^2 \right]$, 
since $\E_g \left[ \norm{g^{(t)}_{(i)}}^2 \right] \le 
\sigma^2 + \norm{\g \mL(w^{(t)})}^2$ 
and $\E [\norm{\xi}^2] = d \sigma_1^2$, 
we have 
\begin{align}
    &\E_{g, \xi} \left[ \sum_{i \in \mH} \norm{g^{(t)}_{(i)}}^2 
+ \sum_{i \in [n]\setminus \mH} \norm{g^{(t)}_{(i)}}^2
+ \norm{\xi}^2  \right] \\
    &\le (n-\nfail) \left( \sigma^2 + \norm{\g \mL(w^{(t)})}^2 \right) + \nfail R + d \sigma_1^2 \\
    &\le (n-\nfail) \left( \sigma^2 + \beta^2 \norm{w^{(t)} - w^*}^2 \right) + \nfail R + d \sigma_1^2 
\end{align}

Plugging in the lower and upper bounds, we have 
\begin{align}
    &\E_{g, \xi} \left[ \norm{w^{(t+1)} - w^*}^2 \right] \\
    &\le 
    \left( 1 - 2\eta \alpha (n-\nfail) + \eta^2 (n+1)(n-\nfail)\beta^2 \right) \norm{w^{(t)} - w^*}^2 \\
    &~~~~+ 2\eta \nfail R \cdot \norm{w^{(t)} - w^*} \\
    &~~~~+ \eta^2 (n+1) \left( (n-\nfail)\sigma^2 + \nfail R + d \sigma_1^2 \right) \\
    &\le \left( 1 - 2\eta \alpha (n-\nfail) + \eta^2 (n^2 + (n+1)(n-\nfail)\beta^2 ) \right) \norm{w^{(t)} - w^*}^2 \\
    &+ \eta^2 (n+1) \left( (n-\nfail)\sigma^2 + \nfail R + d \sigma_1^2 \right) + \left( \frac{\nfail}{n} \right)^2 R^2
\end{align}

Follow similar analysis, but we will have 
\begin{align}
    \rho_A &= 1 - 2\eta \alpha (n-\nfail) + \eta^2 (n^2 + (n+1)(n-\nfail)\beta^2 ) \\
    M_A &= \eta^2 (n+1) \left( (n-\nfail)\sigma^2 + \nfail R + d \sigma_1^2 \right) + \left( \frac{\nfail}{n} \right)^2 R^2
\end{align}
So for every time we run Routine A (with full gradient sum), we have 
\begin{align}
    \E_{g, \xi} \left[ \norm{w^{(t)} - w^*}^2 \right] 
    &\le \rho_A \norm{w^{(t-1)} - w^*}^2 + M_A
\end{align}

And we have $M_B < M_A$, so more routine B can improve the utility.

Overall, we if there are at most $\nfail$ gradients being corrupted at each iteration, we have 
\begin{align}
    \E \left[ \norm{w^{(t)} - w^*}^2 \right] 
    &\le \rho_A \norm{w^{(t-1)} - w^*}^2 + M_A
\end{align}
for 
\begin{align}
    \rho_A &= 1 - 2\eta_A \alpha (n-\nfail) + \eta_A^2 (n^2 + (n+1)(n-\nfail)\beta^2 ) \\
    M_A &= \eta_A^2 (n+1) \left( (n-\nfail)\sigma^2 + \nfail R + d \sigma_1^2 \right) + \left( \frac{\nfail}{n} \right)^2 R^2
\end{align}
if PTR runs Routine A, 
and 
\begin{align}
    \E_{g, \xi} \left[ \norm{w^{(t+1)} - w^*}^2 \right] 
    &\le \rho_B \norm{w^{(t)} - w^*}^2 + M_B
\end{align}
for 
\begin{align}
    \rho_B &= 1 - 2\eta_B \alpha (n-2\nfail) + \eta_B^2 (n^2 + (n-\nfail+1)(n-\nfail)\beta^2 ) \\
    M_B &= \eta_B^2 (n-\nfail+1) \left( (n-\nfail)\sigma^2 + d \sigma_2^2 \right) + \left( \frac{\nfail}{n} \right)^2 R^2
\end{align}
if PTR runs Routine B. 

If we set $\eta_A = \frac{n-F}{n} \eta_B$, 
since 
\begin{align}
    \sigma_2^2 
    \ge \frac{(n-\nfail)(n+1)}{n^2} \left( \frac{(n-\nfail)\sigma^2+FR}{d} + \sigma_1^2 \right)
\end{align}

we have $M_A \le M_B$, and we can also easily verify that $\rho_A \le \rho_B$. 
Thus, as $t \rightarrow \infty$, we have 
\begin{align}
    \E_{g, \xi} \left[ \norm{w^{(t)} - w^*}^2 \right] 
    &\le \frac{M_B}{1-\rho_B}
\end{align}
for 
\begin{align}
    0 < \eta_B \le \frac{ 2 \alpha (n-2\nfail) }{ n^2 + (n-\nfail+1)(n-\nfail)\beta^2 }
\end{align}

\end{proof}


\subsection{Pseudo-code of TSGD+PTR}

The pseudo-code of TSGD+PTR is shown in Algorithm \ref{alg:dpsgd-with-ptr-tmean}, which uses Algorithm \ref{alg:ptr-tmean} (\texttt{PTR-TMEAN}) as a subroutine.  

\newcommand{\loss}{\mathcal{L}}
\newcommand{\z}{\textbf{z}}

\begin{algorithm}[h]
\SetAlgoLined
\SetKwInOut{Input}{input}
\SetKwInOut{Output}{output}
\Input{Dataset $\{\z^1, \ldots, \z^N\}$, loss function $\loss(\theta)=\frac{1}{N} \sum_i \loss(\theta, \z^i)$, learning rate $\eta$, batch size $B$, sensitivity bound $\tau$, Clipping threshold $R$, noise multiplier $\sigma$.}

Initialize $\theta_0$ randomly. 

\For{$t \in [T]$}{

\textbf{Random Subsampling.}

Take a random batch $\B_t$ with sampling probability $q$ in Poisson subsampling. 

\textbf{Obtain Gradients.}

For each $i \in \B_t$, get (potentially faulty) $g^{(t)}_i$. 

\textbf{Gradient Clipping. }

$g^{(t)}_i \leftarrow C \cdot g^{(t)}_i$ for $C = \min \left(1, R / \norm{g^{(t)}_i}_2 \right)$. 

\textbf{Noisy Gradient Aggregation with PTR. }

$\Tilde{g}^{(t)} \leftarrow \texttt{PTR-TMEAN} \left(\{g^{(t)}_i\} \right)$. 

\textbf{Descent.}

$\theta_{t+1}, \omega \leftarrow \theta_{t} - \eta \Tilde{g}^{(t)}$. 

\textbf{Adjust $\nfail$.}

\If{$\omega$ is `+'}{ Increase $\nfail$. }
\Else{ Decrease $\nfail$. }

}

\caption{Private Trimmed-mean SGD with Propose-Test-Release}
\label{alg:dpsgd-with-ptr-tmean}
\end{algorithm}

\begin{algorithm}[h]
\SetAlgoLined
\SetKwInOut{Input}{input}
\SetKwInOut{Output}{output}
\Input{
$S$ -- Set of (clipped) gradient vectors at step $t$: $\{g^{(t)}_i\} \subseteq \R^d$, 
}

$\Delta \leftarrow \min_{\Tilde{S} \in \{ \Tilde{S}: \localSen_{f_2}( \Tilde{S} ) > \tau \} } d \left(S, \Tilde{S} \right)$. 

$\widehat \Delta \leftarrow \Delta + \lap(0, b)$. 

\If{$\widehat \Delta \le \log(1/(2\delta_0)) b$}{
    \Return{$\mean(S) + \sigma R \cdot \N (0, \iden_d)$, `$+$'}
}\Else{
    \Return{$\tmean_\nfail(S) + \sigma \tau \cdot \N (0, \iden_d)$, `$-$'}
}

\caption{\texttt{PTR-TMEAN}}
\label{alg:ptr-tmean}
\end{algorithm}

\subsubsection{Why not directly apply PTR to regular SGD?}

As we discussed in the main text, PTR typically works with robust statistics such as trimmed mean. Regular SGD use mean as gradient aggregation function. Mean, however, does not have a low local sensitivity on most of the inputs. 
Therefore, we focus on the application of PTR in privatizing robust statistics.

\newpage

\section{Experiment Settings \& Additional Results}
\label{appendix:experiment}

\subsection{Experiment Settings for Table \ref{tb:full}}

\subsubsection{Corruption Simulation.}

\rebuttal{
Following the literature in Byzantine robustness \cite{yin2018byzantine,xie2019zeno,acharya2021robust,gupta2021byzantine}, we consider three possible sources of Byzantine failures: corruption in features, labels and communicated gradients. All experiments are repeated for 0\% (i.e., clean), 10\%, and 20\% corruption ratio (CR). 

\textbf{Feature Corruption.} Corruption in Features can arise from the process of data collection. 
Following \cite{acharya2021robust}, we adopt the additive corruption introduced in \cite{hendrycks2018benchmarking}. Specifically, we add Gaussian noise from $\N(0, 100)$ directly to the corrupted images. 

\textbf{Gradient Corruption.} 
Gradient can be corrupted in distributed SGD, e.g., due to hardware malfunction or malicious users. 
We consider the gradient corruption following \cite{xie2019zeno, acharya2021robust}, where we add Gaussian noise from $\N(0, 100)$ to the true gradients. 


\textbf{Label Corruption.} 
Noisy labels are pervasive in the dataset. We randomly flip of label of certain amount of data points.
}

\subsubsection{Datasets \& Models. }
MNIST~\cite{lecun1998mnist} is one of the most commonly used benchmark datasets in deep learning containing 70000 handwritten digit images. CIFAR-10~\cite{cifar} is another classic benchmark for image classification. It consists of 60000 images from 10 different classes with 6000 images each. EMNIST \cite{cohen2017emnist} is similar to MNIST but has a much larger size (145,600 character images and 26 balanced classes). 

In Table \ref{tb:full}, all models are trained entirely from scratch. 
For all datasets, we use a small CNN whose architecture is inherited from the official tutorial of tensorflow/privacy\footnote{https://github.com/tensorflow/privacy}. 

\newcommand{\tsgdptr}{\texttt{TSGD+PTR}}
\newcommand{\tsgdgau}{\texttt{TSGD+Gaussian}}

\subsubsection{Hyperparameters. }
For $\tsgdptr$, we set $\delta_0 = 10^{-8}, b=1$. 
For MNIST and EMNIST, we set gradient clipping bound $R=1$, $\tau=0.5$, noise multiplier $\sigma=1.1$ for $\tsgdptr$, and $\sigma=0.7$ for $\tsgdgau$. The noise multiplier for $\tsgdptr$ and $\tsgdgau$ are picked differently in order to align their privacy loss in each iteration. 
For CIFAR10, we set gradient clipping bound $R=3$, $\tau=2$, noise multiplier $\sigma=1.1$ for $\tsgdptr$, and $\sigma=0.9$ for $\tsgdgau$. 

For MNIST and EMNIST dataset, we set the learning rate as 0.15, batch size as 256; for CIFAR10 dataset, we set the learning rate as 0.1, batch size as 1024. 

We set $\nfail$ to be 25\% of the batch size for $\tsgdgau$. For $\tsgdptr$, $\nfail$ is dynamically adjusted based on the value of $\widehat \Delta$. If sensitivity test is passed, we increase $\nfail$ by $0.02 \times \mathrm{batch size}$, and if sensitivity test is failed, we decrease $\nfail$ by the same amount. 

All of our experiments are performed on Tesla P100-PCIE-16GB GPU.

\subsection{Additional Results on More Architectures}

We experiment with more architectures on CIFAR10 dataset. 
Specifically, we use two famous, moderately large architecture ResNet18 \cite{he2016deep} and VGG11 \cite{simonyan2014very}. We follow the common procedure in prior works \cite{DBLP:conf/ccs/AbadiCGMMT016}: we use ResNet18 and VGG11 that are pretrained by ImageNet dataset. The pre-training weight is publicly available from PyTorch. We only finetune the last layer of the model. 

We set gradient clipping bound $R=5$, batch size as 2048, learning rate 0.01. For $\tsgdptr$, we set $\delta_0 = 10^{-8}, b=1$, and $\tau=3$. We set noise multiplier $\sigma=2.2$ for $\tsgdptr$, and $\sigma=1.8$ for $\tsgdgau$. 
We set $\nfail$ to be 25\% of the batch size for $\tsgdgau$. $\nfail$ is dynamically adjusted in the same way as the experiments in the main text. 

The results are shown in Table \ref{tb:full-appendix}. 
As we can see, $\tsgdptr$ consistently outperforms $\tsgdgau$ across different architectures. 




\begin{table}[t]
\centering
\resizebox{\columnwidth}{!}{
\begin{tabular}{@{}ccccccc@{}}
\toprule
\multirow{2}{*}{\textbf{Archi.}}   & \multirow{2}{*}{\textbf{\begin{tabular}[c]{@{}c@{}}Corruption\\ Type\end{tabular}}} & \multirow{2}{*}{\textbf{CR}} & \multicolumn{2}{c}{$\eps = 3.0$}                               & \multicolumn{2}{c}{$\eps=5.0$}                                \\
                                   &                                                                                     &                              & \textbf{TSGD+Gaussian} & \textbf{TSGD+PTR}                     & \textbf{TSGD+Gaussian} & \textbf{TSGD+PTR}                    \\ \midrule
\textbf{}                          & \textbf{}                                                                           & \textbf{0}                   & 50.05\%                & 52.09\% $\textcolor{red}{(+2.04\%)}$  & 51.75\%                & 52.85\% $\textcolor{red}{(+1.1\%)}$  \\
\multirow{6}{*}{\textbf{VGG11}}    & \multirow{2}{*}{\textbf{Label}}                                                     & \textbf{0.1}                 & 44.03\%                & 48.73\% $\textcolor{red}{(+4.7\%)}$   & 48.78\%                & 50.24\% $\textcolor{red}{(+1.46\%)}$ \\
                                   &                                                                                     & \textbf{0.2}                 & 35.04\%                & 43.63\% $\textcolor{red}{(+8.59\%)}$  & 43.17\%                & 46.35\% $\textcolor{red}{(+3.18\%)}$ \\
                                   & \multirow{2}{*}{\textbf{Feature}}                                                   & \textbf{0.1}                 & 45.59\%                & 49.57\% $\textcolor{red}{(+3.98\%)}$  & 49.60\%                & 50.74\% $\textcolor{red}{(+1.14\%)}$ \\
                                   &                                                                                     & \textbf{0.2}                 & 43.82\%                & 47.95\% $\textcolor{red}{(+4.13\%)}$  & 48.08\%                & 48.88\% $\textcolor{red}{(+0.8\%)}$  \\
                                   & \multirow{2}{*}{\textbf{Gradient}}                                                  & \textbf{0.1}                 & 45.61\%                & 50.15\% $\textcolor{red}{(+4.54\%)}$  & 50.10\%                & 51.15\% $\textcolor{red}{(+1.05\%)}$ \\
                                   &                                                                                     & \textbf{0.2}                 & 45.40\%                & 50.15\% $\textcolor{red}{(+4.75\%)}$  & 50.46\%                & 50.82\% $\textcolor{red}{(+0.36\%)}$ \\ \midrule
\multirow{2}{*}{}                  & \multirow{2}{*}{}                                                                   & \multirow{2}{*}{}            & \multicolumn{2}{c}{$\eps = 3.0$}                               & \multicolumn{2}{c}{$\eps=5.0$}                                \\
                                   &                                                                                     &                              & \textbf{TSGD+Gaussian} & \textbf{TSGD+PTR}                     & \textbf{TSGD+Gaussian} & \textbf{TSGD+PTR}                    \\  \midrule
                                   &                                                                                     & \textbf{0}                   & 38.97\%                & 43.15\% $\textcolor{red}{(+3.18\%)}$  & 46.27\%                & 47.15\% $\textcolor{red}{(+0.88\%)}$ \\
\multirow{6}{*}{\textbf{ResNet18}} & \multirow{2}{*}{\textbf{Label}}                                                     & \textbf{0.1}                 & 35.84\%                & 42.55\% $\textcolor{red}{(+6.71\%)}$  & 42.05\%                & 44.77\% $\textcolor{red}{(+2.72\%)}$ \\
                                   &                                                                                     & \textbf{0.2}                 & 26.15\%                & 36.48\% $\textcolor{red}{(+10.33\%)}$ & 33.97\%                & 39.71\% $\textcolor{red}{(+5.74\%)}$ \\
                                   & \multirow{2}{*}{\textbf{Feature}}                                                   & \textbf{0.1}                 & 37.91\%                & 43.15\% $\textcolor{red}{(+5.24\%)}$  & 42.54\%                & 44.93\% $\textcolor{red}{(+2.39\%)}$ \\
                                   &                                                                                     & \textbf{0.2}                 & 36.30\%                & 41.85\% $\textcolor{red}{(+5.55\%)}$  & 41.21\%                & 43.84\% $\textcolor{red}{(+2.63\%)}$ \\
                                   & \multirow{2}{*}{\textbf{Gradient}}                                                  & \textbf{0.1}                 & 38.42\%                & 44.27\% $\textcolor{red}{(+5.85\%)}$  & 43.83\%                & 46.4\% $\textcolor{red}{(+2.57\%)}$  \\
                                   &                                                                                     & \textbf{0.2}                 & 37.79\%                & 44.34\% $\textcolor{red}{(+6.55\%)}$  & 43.22\%                & 46.14\% $\textcolor{red}{(+2.92\%)}$ \\ 
                                   \bottomrule
\end{tabular}
}
\caption{
Model Accuracy under different privacy budgets and corruption settings. Every statistic is averaged over 5 runs with different random seed. The improvement of $\tsgd+\ptr$ over $\tsgd+\gau$ is highlighted in the red text. 
}
\label{tb:full-appendix}
\end{table}

\subsection{Additional Results on More Corruption Types}

\rebuttal{
Besides the three corruption types we considered in the maintext, we evaluate on two additional possible errors which are considered more severe kinds of failure. 
\begin{enumerate}
    \item \textbf{Gradient Bit-flipping failure} where the bits that control the sign of the floating numbers are flipped, e.g., due to some hardware failure. A faulty worker pushes the negative gradient instead of the true gradient to the servers. \item \textbf{Targeted label flipping failure} where the labels are flipped in a ``targeted'' way, i.e., for any \emph{label} $\in \{0, \ldots, 25\}$, is replaced by $25-$\emph{label}. Such failures/attacks can be caused by data poisoning or software failures. 
\end{enumerate}
We experiment on EMNIST dataset and the results are shown in Table \ref{tb:more-attack}. 
The experiment settings are exactly the same as the settings for Table \ref{tb:full}. As we can see, $\tsgdptr$ once again outperform $\tsgdgau$ significantly. 
}

\begin{table}[t]
\centering
\resizebox{\columnwidth}{!}{
\begin{tabular}{@{}cccccc@{}}
\toprule
\textbf{Corruption}               & \multirow{2}{*}{\textbf{CR}} & \multicolumn{2}{c}{$\eps = 2.0$}                              & \multicolumn{2}{c}{$\eps=2.5$}                                \\ 
\textbf{Type}                     &                              & \textbf{TSGD+Gaussian} & \textbf{TSGD+PTR}                    & \textbf{TSGD+Gaussian} & \textbf{TSGD+PTR}                    \\ \midrule
\textbf{}                         & \textbf{0}                   & 72.94\%                & 76.44\% $\textcolor{red}{(+3.5\%)}$  & 79.15\%                & 81.02\% $\textcolor{red}{(+1.87\%)}$ \\
\multirow{2}{*}{\textbf{Targeted Label Flip}}   & \textbf{0.1}                 & 72.06\%                & 75.11\% $\textcolor{red}{(+3.05\%)}$ & 76.25\%                & 78.97\% $\textcolor{red}{(+2.72\%)}$ \\
                                  & \textbf{0.2}                 & 70.72\%                & 73.85\% $\textcolor{red}{(+3.13\%)}$ & 73.69\%                & 76.67\% $\textcolor{red}{(+2.98\%)}$ \\
\multirow{2}{*}{\textbf{Gradient Bit Flip}} & \textbf{0.1}                 & 69.58\%                & 75.67\% $\textcolor{red}{(+6.09\%)}$  & 75.39\%                & 78.82\% $\textcolor{red}{(+3.43\%)}$ \\
                                  & \textbf{0.2}                 & 65.13\%                & 75.67\% $\textcolor{red}{(+10.54\%)}$ & 71.85\%                & 78.85\% $\textcolor{red}{(+7.0\%)}$  \\ \bottomrule
\end{tabular}
}
\caption{
Model Accuracy on EMNIST dataset under different privacy budgets on two more severe types of failures. 
}
\label{tb:more-attack}
\end{table}

\subsection{Comparison between regular DPSGD and trimmed-mean based robust SGD}

\rebuttal{
We additionally show the comparison between trimmed mean robust SGD with/without PTR and regular DPSGD in Table \ref{tb:compare-with-dpsgd} on EMNIST dataset with the same experiment settings described before. 
As we can see, the robust SGD performs worse than non-robust counterpart on clean training data. This is because when the training data are clean, the outliers filtered out by robust SGD in the gradient batch are usually corresponding to the data points that are misclassified, which are important for improving model performance. However, $\tsgdptr$ achieves better performance on most of the corruption settings. 
}

\begin{table}[t]
\centering
\resizebox{\columnwidth}{!}{
\begin{tabular}{@{}cccccccc@{}}
\toprule
\textbf{Corruption}                & \multirow{2}{*}{\textbf{CR}} & \multicolumn{3}{c}{$\eps = 2.0$}                                                & \multicolumn{3}{c}{$\eps=2.5$}                                                  \\
\textbf{Type}                      &                              & \textbf{TSGD+Gaussian} & \textbf{DPSGD} & \textbf{TSGD+PTR}                     & \textbf{TSGD+Gaussian} & \textbf{DPSGD} & \textbf{TSGD+PTR}                     \\ \midrule
\textbf{}                          & \textbf{0}                   & 72.94\%                & 77.29\%        & 76.44\% $\textcolor{blue}{(-0.85\%)}$ & 79.15\%                & 83.06\%        & 81.02\% $\textcolor{blue}{(-2.04\%)}$ \\
\multirow{2}{*}{\textbf{Label}}    & \textbf{0.1}                 & 72.60\%                & 74.80\%        & 75.66\% $\textcolor{red}{(+0.86\%)}$  & 79.38\%                & 79.30\%        & 80.63\% $\textcolor{red}{(+1.33\%)}$  \\
                                   & \textbf{0.2}                 & 70.03\%                & 71.42\%        & 72.62\% $\textcolor{red}{(+1.2\%)}$   & 77.48\%                & 77.43\%        & 79.19\% $\textcolor{red}{(+1.76\%)}$  \\
\multirow{2}{*}{\textbf{Feature}}  & \textbf{0.1}                 & 69.60\%                & 74.58\%        & 74.01\% $\textcolor{blue}{(-0.57\%)}$  & 77.80\%                & 80.95\%        & 81.04\% $\textcolor{red}{(+0.09\%)}$  \\
                                   & \textbf{0.2}                 & 70.04\%                & 73.79\%        & 74.76\% $\textcolor{red}{(+0.97\%)}$  & 78.99\%                & 79.43\%        & 80.74\% $\textcolor{red}{(+1.31\%)}$  \\
\multirow{2}{*}{\textbf{Gradient}} & \textbf{0.1}                 & 71.75\%                & 72.97\%        & 76.19\% $\textcolor{red}{(+3.22\%)}$  & 77.32\%                & 76.59\%        & 77.73\% $\textcolor{red}{(+1.14\%)}$  \\
                                   & \textbf{0.2}                 & 70.76\%                & 71.22\%        & 74.65\% $\textcolor{red}{(+3.43\%)}$  & 76.68\%                & 75.69\%        & 77.17\% $\textcolor{red}{(+1.48\%)}$ \\ \bottomrule
\end{tabular}
}
\caption{
\rebuttal{
Model accuracy comparison with regular DPSGD on EMNIST dataset. 
The improvement of $\tsgd+\ptr$ over regular DPSGD is highlighted in the red text. }
}
\label{tb:compare-with-dpsgd}
\end{table}

\subsection{Additional Results on Privacy Analysis Comparison}
\rebuttal{
In this section, we show more numerical results on privacy analysis comparison by varying $\tau = \sigma_2 / \sigma_1$. 
In Figure \ref{fig:rdp-vs-dp-value-more}, we numerically compute the privacy bound from direct analysis and the one converted from RDP, with $\tau \in [2/3, 1/3, 0.1]$. 
In Figure \ref{fig:rdp-compare-more}, we show the subsampled privacy bound composed with moment account, also with $\tau \in [2/3, 1/3, 0.1]$. 
As we can see, our Theorem \ref{thm:rdp-for-ptr-maintext} and Theorem \ref{thm:subsampled-ptr} once again provide tighter bounds compared with the baseline. 
}

\begin{figure}[t]
    \centering
    \includegraphics[width=\columnwidth]{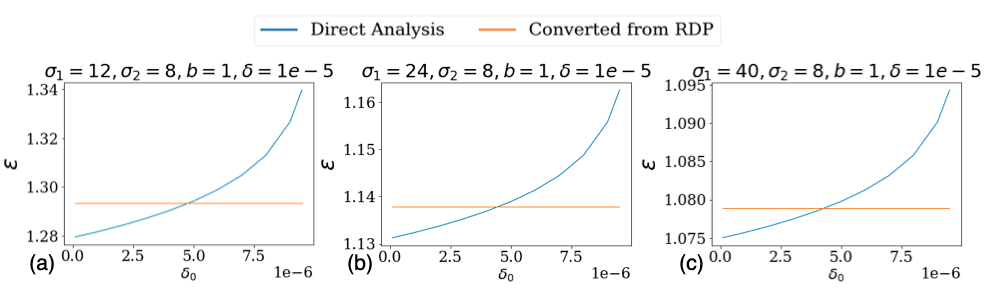}
    \caption{
    The $\eps$ parameter of the $(\eps, \delta)$-DP guarantee of PTR when $\delta = 10^{-5}$ for different noise scales. We convert the RDP bound in Theorem \ref{thm:rdp-for-ptr-maintext} to $(\eps, \delta)$-DP by the RDP-DP conversion formula from \cite{balle2020hypothesis}, and compare it with the $\eps$ obtained from the direct analysis in Theorem \ref{thm:dp-for-ptr-maintext}. For the bound converted from RDP, we search for the optimal $\alpha \in [1, 200]$. The bound is constant across different $\delta_0$ since when $\delta_0$ is small, the RDP for PTR will take the second term in (\ref{eq:ptr-rdp-bound}). 
    }
    \label{fig:rdp-vs-dp-value-more}
\end{figure}

\begin{figure}[t]
    \centering
    \includegraphics[width=\columnwidth]{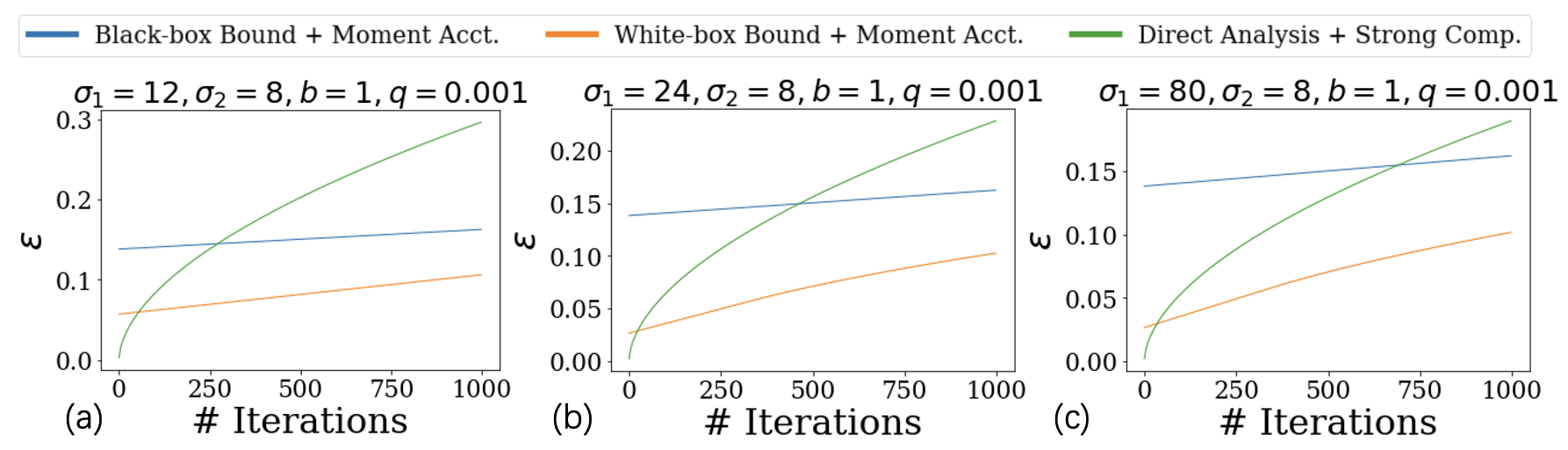}
    \caption{
    Illustration of the use of our Theorem \ref{thm:subsampled-ptr} in moments accountant. We plot the the privacy loss $\eps$ for $\delta=10^{-5}$ after different rounds of composition. We set $\delta_0 = 10^{-8}$ here to allow more iterations for Strong Composition of $(\eps, \delta)$-DP. 
    }
    \label{fig:rdp-compare-more}
\end{figure}

\end{document}